%% file: main.tex
\documentclass[11pt]{article}
\usepackage[margin=1in]{geometry}
\usepackage{graphicx}
\usepackage{float}
\usepackage[utf8]{inputenc}
\usepackage[english]{babel}
\usepackage{amsmath,amsthm}
\usepackage{amssymb}
\usepackage{amsfonts}
\usepackage{CJK}

\usepackage{hyperref}
\usepackage{algorithm,algorithmic}
\usepackage{bbm}

\usepackage{microtype}%if unwanted, comment out or use option "draft"
\allowdisplaybreaks

\title{Synchronization Strings: Efficient and Fast Deterministic Constructions over Small Alphabets\footnote{Supported in part by NSF grants CCF-1527110, CCF-1618280, CCF-1617713 and NSF CAREER award CCF-1750808.}}

\author{Kuan Cheng\thanks{kcheng17@jhu.edu. Department of Computer Science, Johns Hopkins University.   }
\and Bernhard Haeupler \thanks{haeupler@cs.cmu.edu. Department of Computer Science, Carnegie Mellon University}
 \and Xin Li\thanks{lixints@cs.jhu.edu. Department of Computer Science, Johns Hopkins University. }
 \and Amirbehshad Shahrasbi \thanks{shahrasbi@cs.cmu.edu. Department of Computer Science, Carnegie Mellon University}
  \and Ke Wu\thanks{AshleyMo@jhu.edu. Department of Computer Science, Johns Hopkins University.}
  }

\newcommand{\eps}{\varepsilon}

\newcommand{\N}{\ensuremath{\mathbb N}}

\newcommand{\poly}{\ensuremath{\mathsf{poly}}}

\newtheorem{theorem}{Theorem}[section]
\newtheorem{definition}[theorem]{Definition}
\newtheorem{corollary}{Corollary}[theorem]
\newtheorem{lemma}[theorem]{Lemma}

\theoremstyle{plain}

\newtheorem*{claim*}{Claim}
\newcommand{\FullOrShort}{full}

\ifthenelse{\equal{\FullOrShort}{full}}{
	\newcommand{\fullOnly}[1]{#1}
  	\newcommand{\shortOnly}[1]{}
	
  }{
    \newcommand{\fullOnly}[1]{}
	\newcommand{\shortOnly}[1]{#1}
   
  }

\begin{document}

\maketitle

\begin{abstract}
\input{abstract.tex}
\end{abstract}

\newpage
\input{intro.tex}
\input{preliminary.tex}

\input{synccircle.tex}

\input{clong.tex}

\input{constalphabet}

%\paragraph{Organization.} All detailed proofs appear in the appendix. Specifically, proofs in Section \ref{sec:synccircle} appear in Appendix A, proofs in Section \ref{sec:LongDistSync} appear in Appendix B, and proofs in Section \ref{sec:syncexsmallalphabets} appear in Appendix C.

\section*{Organization and Acknowledgements}
All detailed proofs appear in the appendix.\ Specifically, proofs in Section \ref{sec:synccircle} appear in Appendix A, proofs in Section \ref{sec:LongDistSync} appear in Appendix B, and proofs in Section \ref{sec:syncexsmallalphabets} appear in Appendix C.

The authors thank Raymond Kang for valuable discussions in the early stages of this work, contributions to Theorem~\ref{thm:MorphismsImpossibility}, and experiments in Appendix~\ref{sec:epsForAlphabet}. We also thank Noga Alon for referring us to the previous work on the twin word problem.

%%
%% Bibliography
%%

%% Either use bibtex (recommended),

\bibliography{reference}

%% .. or use the thebibliography environment explicitely
\newpage

\input{appendix.tex}
\end{document}

%% file: abstract.tex
Synchronization strings are recently introduced by Haeupler and Shahrasbi \cite{haeupler2017synchronization} in the study of codes for correcting insertion and deletion errors (insdel codes).\ A synchronization string is an encoding of the indices of the symbols in a string, and together with an appropriate decoding algorithm it can transform insertion and deletion errors into standard symbol erasures and corruptions. This reduces the problem of constructing insdel codes to the problem of constructing standard error correcting codes, which is much better understood. Besides this, synchronization strings are also useful in other applications such as synchronization sequences and interactive coding schemes. For all such applications, synchronization strings are desired to be over alphabets that are as small as possible, since a larger alphabet size corresponds to more redundant information added.

Haeupler and Shahrasbi \cite{haeupler2017synchronization} showed that for any  parameter $\varepsilon>0$, synchronization strings of arbitrary length exist over an alphabet whose size depends only on $\varepsilon$. Specifically, \cite{haeupler2017synchronization} obtained an alphabet size of $O(\varepsilon^{-4})$, which left an open question on where the minimal size of such alphabets lies between $\Omega(\varepsilon^{-1})$ and $O(\varepsilon^{-4})$. In this work, we partially bridge this gap by providing an improved lower bound of $\Omega\left(\eps^{-3/2}\right)$, and an improved upper bound of $O\left(\eps^{-2}\right)$. We also provide fast explicit constructions of synchronization strings over small alphabets.

Further, along the lines of previous work on similar combinatorial objects, we study the extremal question of the smallest possible alphabet size over which synchronization strings can exist for some constant $\varepsilon < 1$. We show that one can construct $\varepsilon$-synchronization strings over alphabets of size four while no such string exists over binary alphabets. This reduces the extremal question to whether synchronization strings exist over ternary alphabets.

%Amazingly, Haeupler and Shahrasbi \cite{haeupler2017synchronization} showed that for any error parameter $\varepsilon>0$, synchronization strings of arbitrary length exist over an alphabet whose size depends only on $\varepsilon$. Specifically, \cite{haeupler2017synchronization} obtained an alphabet size of $O(\varepsilon^{-4})$, as well as a randomized construction that runs in expected time $O(n^5)$. However, it remains an interesting question to find deterministic and more efficient constructions.

%In this paper, we improve the construction in \cite{haeupler2017synchronization} in three aspects: we achieve a smaller alphabet size, a deterministic construction, and a faster algorithm. Along the way we introduce a new combinatorial object, and establish a new connection between synchronization strings and insdel codes --- such codes can be used in a simple way to construct synchronization strings. This new connection complements the connection found in \cite{haeupler2017synchronization}, and may be of independent interest. In an independent work \cite{HS17c}, Haeupler and Shahrasbi also give deterministic constructions of synchronization strings over arbitrary length (or even infinite length). Their constructions can achieve linear construction time, but have alphabet size $\varepsilon^{-O(1)}$, which may be larger than ours.

%% file: intro.tex
\section{Introduction}
This paper focuses on the study of a combinatorial object called \emph{synchronization string}. Intuitively, a synchronization string is a (finite or infinite) string that avoids similarities between pairs of intervals in the string. Such nice properties and synchronization strings themselves can actually be motivated from at least two different aspects: coding theory and pattern avoidance. We now discuss the motivations and previous work in each aspect below.

\subsection{Motivation and Previous Work in Coding Theory}
The general and most important goal of the coding theory is to obtain a reliable transmission of information in the presence of noise or adversarial error. Starting from the pioneering works of Shannon, Hamming, and many others, coding theory has evolved into an extensively studied field, with applications found in various areas in computer science. Regarding the general goal of correcting errors, we now have a very sophisticated and almost complete understanding of how to deal with symbol erasures and corruptions. On the other hand, the knowledge of codes for synchronization errors such as insertions and deletions, has lagged far behind despite also being studied intensively since the 1960s. In practice, this is one of the main reasons why communication systems require a lot of effort and resources to maintain synchronization strictly.

One major difficulty in designing codes for insertion and deletion errors is that in the received codeword, the positions of the symbols may have changed. This is in contrast to standard symbol erasures and corruptions, where the positions of the symbols always stay the same. Thus, many of the known techniques in designing codes for standard symbol erasures and corruptions cannot be directly utilized to protect against insertion and deletion errors. %Naturally, if one can find a way to bridge this gap and transform insertion and deletion errors into symbol erasures and corruptions, this will make our life much easier. 

In this context, a recent work of Haeupler and Shahrasbi \cite{haeupler2017synchronization} introduced \emph{synchronization strings}, which enable a black-box transformation of Hamming-type error correcting codes to codes that protect against insertions and deletions.
Informally, a synchronization string of length $n$ is an encoding of the indices of the $n$ positions into a string over some alphabet $\Sigma$, such that, despite some insertion and deletion errors, one can still recover the correct indices of many symbols. With the correct indices of these symbols, a standard error correcting code can then be used to recover the original message. This then gives a code for insertion and deletion errors, which is the combination of a standard error correcting code and a synchronization string.

The simplest example of a synchronization string is just to record the index of each symbol, i.e, the string $1, 2, \cdots, n$. It can be easily checked that even if $(1-\varepsilon)$ fraction of these indices are deleted, one can still correctly recover the positions of the remaining $\varepsilon n$ symbols. However, this  synchronization string uses an alphabet whose size grows with the length of the string. The main contribution of \cite{haeupler2017synchronization} is to show that under a slight relaxation, there exist synchronization strings of arbitrary length $n$ over an alphabet with fixed size.
Further, \cite{haeupler2017synchronization} provided efficient and streaming methods to correctly recover the indices of many symbols from a synchronization string after being altered by insertion and deletion errors. Formally, \cite{haeupler2017synchronization} defines \emph{$\eps$-synchronization strings} as follows. A string $S$ is an $\eps$-synchronization string if the edit distance of any two consecutive substrings $S[i, j)$ and $S[j ,k)$ is at least $(1-\eps)(k-i)$.

%\begin{definition} \cite{haeupler2017synchronization} \label{def:sc} ($\varepsilon$-synchronization string)  For some alphabet $\Sigma$, a string $S \in \Sigma^n$ is an $\varepsilon$-synchronization string if $\forall 1\leq i < j <k\leq n$, we have that $ED(S[i,j],S[j+1,k])>(1-\varepsilon)(k-i)$ where $ED( , )$ stands for the edit distance of two strings, and $S[i,j]$ means the continuous subsequence of $S$ from $i$th position to $j$th position, both ends included.
%\end{definition}
 
Using the construction and decoding methods for $\eps$-synchronization strings, \cite{haeupler2017synchronization} gives a code that for any $\delta \in (0,1)$ and $\varepsilon>0$, can correct $\delta$ fraction of insertion and deletion errors with rate $1-\delta-\varepsilon$. % and alphabet size $\varepsilon^{-O(1)}$.%Together this gives a code that can correct $1-\varepsilon$ fraction of insertion and deletion errors, with rate $\Omega(\varepsilon)$ and alphabet size $\varepsilon^{-O(1)}$. 
Besides this, synchronization strings have found a variety of applications, such as in  synchronization sequences \cite{mercier2010survey}, interactive coding schemes \cite{gelles2015coding, gelles2015capacity, ghaffari2014optimal1, ghaffari2014optimal2, haeupler2014interactive, kol2013interactive, haeupler2017synsimucode, HS17c}, coding against synchronization errors~\cite{HS17c, haeupler2018synchronization}, and edit distance tree codes \cite{braverman2017coding, haeupler2017synsimucode}. %Furthermore, because of the nice properties of synchronization strings, it is plausible that they will find other applications in the future. However, despite the usefulness of such objects, it remains an interesting open problem to find deterministic and more efficient constructions of synchronization strings, as in \cite{haeupler2017synchronization} the authors only give a randomized construction.

For all such applications, synchronization strings are desired to be over alphabets that are as small as possible, since a larger alphabet size corresponds to more redundant information added.\ Thus a natural question here is how small the alphabet size can be. In \cite{haeupler2017synchronization}, Haeupler and Shahrasbi showed that $\varepsilon$-synchronization strings with arbitrary length exist over an alphabet of size $O(\varepsilon^{-4})$, they also gave a randomized polynomial time algorithm to construct such strings. In a very recent work \cite{HS17c}, they further gave various efficient deterministic constructions for finite/infinite $\varepsilon$-synchronization strings, which have alphabets of size $\poly\left(\eps^{-1}\right)$ for some unspecified large polynomials.\ On the other hand, the definition of synchronization strings implies that any $\eps^{-1}$ consecutive symbols in an $\eps$-synchronization string have to be distinct\textemdash providing an $\Omega\left(\eps^{-1}\right)$ lower-bound for the alphabet size.%This motivates the main purpose of this work which is to find synchronization strings over small alphabets.

\subsection{Motivation and Previous Work in Pattern Avoidance}
Apart from applications in coding theory and other communication problems involving insertions and deletions, synchronization strings are also interesting combinatorial objects from a mathematical perspective. As a matter of fact, plenty of very similar combinatorial objects have been studied prior to this work.

A classical work of Axel Thue~\cite{thue1977unendliche} introduces and studies \emph{square-free} strings, i.e., strings that do not contain two identical consecutive substrings. Thue shows that such strings exist over alphabets of size three and provides a fast construction of such strings using \emph{morphisms}. The seminal work of Thue inspired further works on the same problem~\cite{thue1912gegenseitige, leech19572726, crochemore1982sharp, shelton1981aperiodic, shelton1982aperiodic, zolotov2015another} and problems with a similar pattern avoidance theme.

Krieger et.\ al.\ ~\cite{krieger2007avoiding} study strings that satisfy relaxed variants of square-freeness, i.e., strings that avoid \emph{approximate squares}. Their study provides several results on strings that avoid consecutive substrings of equal length with small additive or multiplicative Hamming distance in terms of their length. In each of these regimes, \cite{krieger2007avoiding} gives constructions of approximate square free strings over alphabets with small constant size for different parameters. %several values of the additive or multiplicative approximation parameter.

Finally, Camungol and Rampersad~\cite{camungol2016avoiding} study \emph{approximate squares with respect to edit distance}, which is equivalent to the $\eps$-synchronization string notion except that the edit distance property is only required to hold for pairs of consecutive substrings of equal length. \cite{camungol2016avoiding} employs a technique based on entropy compression to prove that such strings exist over alphabets that are constant in terms of string length but exponentially large in terms of $\eps^{-1}$. We note that the previous result of Haeupler and Shahrasbi \cite{haeupler2017synchronization} already improves this dependence to $O(\varepsilon^{-4})$.

Again, a main question addressed in most of the above-mentioned previous work on similar mathematical objects is how small the alphabet size can be.

\subsection{Our Results}
In this paper we study the question of how small the alphabet size of an $\eps$-synchronization string can be. We address this question both for a specified $\eps$ and for unspecified $\eps$. In the first case we try to bridge the gap between the upper bound of $O\left(\eps^{-4}\right)$ provided in \cite{haeupler2017synchronization} and the lower bound of $\Omega\left(\eps^{-1}\right)$. In the second case we study the question of how small the alphabet size can be to ensure the existence of an $\eps$-synchronization string for some constant $\eps<1$. In both cases we also give efficient constructions that improve previous results.

\subsubsection{New Bounds on Minimal Alphabet Size for a given $\eps$}

Our first theorem gives improved upper bound and lower bound for the alphabet size of an $\eps$-synchronization string for a given $\eps$.
\begin{theorem}
For any $0<\eps<1$, there exists an alphabet $\Sigma$ of size $O\left(\eps^{-2}\right)$ such that an infinite $\eps$-synchronization string exists over $\Sigma$. In addition, $\forall n\in \mathbb{N}$, a randomized 
algorithm can construct an $\eps$-synchronization string of length $n$ in expected 
time $O(n^5 \log n)$. Further, the alphabet size of any $\eps$-synchronization string that is long enough in terms of $\eps$ has to be at least $\Omega\left(\eps^{-3/2}\right)$.
\end{theorem}

Next, we provide efficient and even linear-time constructions of $\eps$-synchronization strings over drastically smaller alphabets than the efficient constructions in \cite{HS17c}.
\begin{theorem}
%thm:polyConstruction
For every $n\in \mathbb{N}$ and
any constant $\varepsilon\in (0,1)$,
there is a deterministic  construction of a (long-distance) $\varepsilon$-synchronization string of length $n$ over an alphabet of size $O(\varepsilon^{-2})$ that runs in $\poly(n)$ time. Further, there is a highly-explicit linear time construction of such strings over an alphabet of size $O(\varepsilon^{-3})$.
\end{theorem}

%\subsubsection{Infinite Synchronization String}
Moreover, in Section~\ref{sec:infiniteConstruction}, we present a method to construct infinite synchronization strings using constructions for finite ones that only increases the alphabet size by a constant factor\textemdash as opposed to the construction in~\cite{HS17c} that increases the alphabet size quadratically. %This leads to the following constructions for infinite synchronization strings.
\begin{theorem}
%thm:polyinfinite
For any constant $0<\varepsilon < 1$, there exists an explicit construction of an infinite $\varepsilon$-synchronization string $S$  over an alphabet of size $O(\varepsilon^{-2})$. Further, there exists a highly-explicit construction of an infinite $\varepsilon$-synchronization string $S$ over an alphabet of size $O(\varepsilon^{-3})$ such that for any $i\in \mathbb{N}$, the first $i$ symbols can be computed in $O(i)$ time and $S[i,i+\log i]$ can be computed in $O(\log i)$ time.
\end{theorem}

\subsubsection{Minimal Alphabet Size for Unspecified $\eps$: Three or Four?	}
One interesting question that has been commonly addressed by previous work on similar combinatorial objects is the size of the smallest alphabet over which one can find such objects. Along the lines of \cite{thue1977unendliche, zolotov2015another, krieger2007avoiding, camungol2016avoiding}, we study the existence of synchronization strings over alphabets with minimal constant size.

 It is easy to observe that no such string can exist over a binary alphabet since any binary string of length four either contains two consecutive identical symbols or two consecutive identical substrings of length two. On the other hand, one can extract constructions over constant-sized alphabets from the existence proofs in~\cite{haeupler2017synchronization, HS17c}, but the unspecified constants there would be quite large.  
 In Section~\ref{sec:alphabetFour}, for some $\eps<1$, we provide a construction of arbitrarily long $\eps$-synchronization strings over an alphabet of size four. This narrows down the question to whether such strings exist over alphabets of size three.

To construct such strings, we introduce the notion of \emph{weak synchronization string}, which requires substrings to satisfy a similar property as that of an $\eps$-synchronization string, except that the lower bound on edit distance is rounded down to the nearest integer.\ We show that weak synchronization strings exist over binary alphabets and use one such string to modify a ternary square-free string (\cite{thue1977unendliche}) into a synchronization string over an alphabet of size four.

Finally, in Appendix~\ref{sec:epsForAlphabet}, we provide experimental evidence for the existence of synchronization strings over ternary alphabets by finding lower-bounds for $\eps$ for which $\eps$-synchronization strings over alphabets of size 3, 4, 5, and 6 might exist. Similar experiments have been provided for related combinatorial objects in the previous work~\cite{krieger2007avoiding, camungol2016avoiding}.

\subsubsection{Constructing Synchronization Strings Using Uniform Morphisms}
Morphisms have been widely used in previous work as a tool to construct similar combinatorial objects.  A uniform morphism of rank $r$ over an alphabet $\Sigma$ is a function $\phi:\Sigma\rightarrow\Sigma^r$ that maps any symbol of an alphabet $\Sigma$ to a string of length $r$ over the same alphabet. Using this technique, some similar combinatorial objects in previous work have been constructed by taking a symbol from the alphabet and then repeatedly using an appropriate morphism to replace each symbol with a string ~\cite{zolotov2015another, krieger2007avoiding}. Here we investigate whether such tools can also be utilized to construct synchronization strings. In Section~\ref{sec:morphism}, we show that no such morphism can construct arbitrarily long $\eps$-synchronization strings for any $\eps<1$.

%% file: preliminary.tex
\section{Some Notations and Definitions}
Usually we use $\Sigma$ (probably with some subscripts) to denote the alphabet and $\Sigma^*$ to denote all strings over alphabet $\Sigma$.

\begin{definition}[Subsequence] The subsequence of a string $S$ is any sequence of symbols obtained from $S$ by deleting some symbols. It doesn't have to be continuous.
\end{definition}

\begin{definition}[Edit distance] For every $n\in \mathbb{N}$, the edit distance $ED(S,S')$ between two strings $S, S'\in\Sigma^n$ is the minimum number of insertions and deletions required to transform $S$ into $S'$.
\end{definition}

\begin{definition}[Longest Common Subsequence] For any strings $S, S'$ over $\Sigma$, the longest common subsequence of $S$ and $S'$ is the longest pair of subsequence that are equal as strings. We denote by $LCS(S,S')$ the length of the longest common subsequence of $S$ and $S'$.
\end{definition}

Note that $ED(S,S') = |S|+|S'|-2LCS(S,S')$ where $|S|$ denotes the length of $S$.

\begin{definition}[$\varepsilon$-synchronization string] A string $S$ is
an $\varepsilon$-synchronization string if  $\forall 1\leq i< j < k \leq |S|+1$, $ED(S[i,j),S[j,k))>(1-\varepsilon)(k-i)$.
\end{definition}

\begin{definition}[square-free string] A string $S$ is
a square free string if  $\forall 1\leq i < i+2l \leq |S|+1$, $(S[i,i+l)$ and $S[i+l,i+2l))$ are different as words.
\end{definition}

We also introduce the following generalization of a synchronization string, which will be useful in our deterministic constructions of synchronization strings.

\begin{definition}[$\varepsilon$-synchronization circle] A string $S$ is
an $\varepsilon$-synchronization circle if  $\forall 1\leq i\leq |S|$, $S_i,S_{i+1},\dots,S_|S|,S_1,S_2,\dots,S_{i-1}$ is an $\varepsilon$-synchronization string.
\end{definition}

%% file: synccircle.tex
\section{$\varepsilon$-synchronization Strings and Circles with Alphabet Size $O(\varepsilon^{-2})$}
\label{sec:synccircle}
In this section we show that by using a non-uniform sample space together with the Lov\'asz Local Lemma, we can have a randomized polynomial time construction of an $\varepsilon$-synchronization string with alphabet size $O(\varepsilon^{-2})$.\ We then use this to give a simple construction of an $\varepsilon$-synchronization circle with alphabet size $O(\varepsilon^{-2})$ as well. Although the constructions here are randomized, the parameter $\varepsilon$ can be anything in $(0, 1)$ (even sub-constant), while our deterministic constructions in later sections usually require $\varepsilon$ to be a constant in $(0, 1)$.

%\subsection{Synchronization String}
We first recall the General Lov\'asz Local Lemma.
\begin{lemma}\label{generalLLL}
(General Lov\'asz Local Lemma) Let $A_1,...,A_n$ be a set of bad events. $G(V,E)$ is a dependency graph for this set of events if $V=\{1,\dots,n\}$ and each event $A_i$ is mutually independent of all the events $\{A_j:(i,j)\notin E\}$.

If there exists $x_1,...,x_n\in [0,1)$ such that for all $i$ we have \[\Pr(A_i)\leq x_i\prod_{(i,j)\in E}(1-x_j)\]
Then the probability that none of these events happens is bounded by\[\Pr[\bigwedge_{i=1}^n \bar{A}_i]\geq \prod_{i=1}^n (1-x_i)>0\]
\end{lemma}

Using this lemma, we have the following theorem showing the existence of $\varepsilon$-synchronization strings over an alphabet of size $O(\varepsilon^{-2})$. %The proof applies a non-uniform sample space when using the General Lov\'asz Local Lemma. 
%Namely, fixing an alphabet $\Sigma$ and $t = O(\eps^{-2})$, we randomly sequentially sample each symbol from $\Sigma$ such that it is different from its previous $t-1$ symbols (if there are less than $t-1$ symbols in previous, then it is different from all of them).
% The key observation in the proof is that, for intervals with no overlap, their corresponding events are still independent.

\begin{theorem}
\label{syncStr}
$\forall \varepsilon\in(0,1)$ and $\forall n \in \N$, there exists an $\varepsilon$-synchronization string $S$ of length $n$ over alphabet $\Sigma$ of size $\Theta(\varepsilon^{-2})$.
\end{theorem}
%
%\begin{proof}[Proof Sketch]
%Fix an alphabet $\Sigma$ with size $\Theta(\varepsilon^{-2})$ and a parameter $t = O(\eps^{-2})$. 
%We randomly sample each symbol in the string from $\Sigma$ conditioned on that it is different from the previous $t-1$ symbols. If there are less than $t-1$ previous symbols, then the current symbol is sampled such that it is different from all previous symbols.
%
%
%For any interval $S[i, k]$, we define the event, that the interval is \textit{bad}, to be such that $ED(S[i,j],S[j+1,k]) \leq (1-\varepsilon)(k-i)$, i.e., $LCS(S[i,j],S[j+1,k]) \geq \frac{\varepsilon}{2}(k-i)$.
%
%The probability that an interval of length $l$ is bad can be upper bounded by $C^{-\varepsilon l}$ for some constant $C$. Next, a careful calculation shows that we have the following claim.
%
%\begin{claim*}
%\label{claim:independentintervals}
%The badness of interval $I=S[i,j]$ is mutually independent of the badness of all intervals that do not intersect with $I$.
%\end{claim*}
%
%
%We can now use Lemma \ref{generalLLL} to prove the theorem.
%\end{proof}

\global\def\DetailedProofOfThmSyncStr{ 	% Define detailed proof of theorem
\shortOnly{\begin{proof}[Proof of Theorem~\ref{syncStr}]}
\fullOnly{\begin{proof}}
Suppose $|\Sigma|=c_1\varepsilon^{-2}$ where $c_1$ is a constant. Let $t=c_2\varepsilon^{-2}$ and $0<c_2<c_1$. The sampling algorithm is as follows:
\begin{enumerate}
  \item Randomly pick $t$ different symbols from $\Sigma$ and let them be the first $t$ symbols of $S$. If $t\geq n$, we just pick $n$ different symbols.
  \item For $t+1\leq i\leq n$, we pick the $i$th symbol $S[i]$ uniformly randomly from $\Sigma\setminus\{S[i-1],\dots,S[i-t+1]\}$
\end{enumerate}
Now we prove that there's a positive probability that $S$ contains no \textit{bad} interval $S[i,k]$ which violates the requirement that $ED(S[i,j],S[j+1,k])>(1-\varepsilon)(k-i)$ for any $i<j<k$. This requirement is equivalent to $LCS(S[i,j],S[j+1,k])< \frac{\varepsilon}{2}(k-i)$.

Notice that for $k-i\leq t$, the symbols in $S[i,k]$ are completely distinct. Hence we only need to consider the case where $k-i>t$. First, let's upper bound the probability that an interval is bad:
\begin{align*}
\Pr [\text{interval I of length } l \text{ is bad}]&\leq\binom{l}{\varepsilon l}(|\Sigma|-t)^{-\frac{\varepsilon l}{2}}\\
&\leq\frac{el}{\varepsilon l}^{\varepsilon l} (|\Sigma|-t)^{-\frac{\varepsilon l}{2}}\\
&\leq(\frac{\varepsilon\sqrt{|\Sigma|-t}}{e})^{-\varepsilon l}\\
& = C^{-\varepsilon l}
\end{align*}
The first inequality holds because if the interval is bad, then it has to contain a repeating sequence $a_1a_2\dots a_pa_1a_2\dots a_p$ where $p$ is at least $\frac{\varepsilon l}{2}$. Such sequence can be specified via choosing $\varepsilon l$ positions in the interval and the probability that a given sequence is valid for the string in this construction is at most $(|\Sigma|-t)^{-\frac{\varepsilon l}{2}}$. The second inequality comes from Stirling's inequality.

The inequality above indicates that the probability that an interval of length $l$ is bad can be upper bounded by $C^{-\varepsilon l}$, where $C$ is a constant and can be arbitrarily large by modifying $c_1$ and $c_2$.

Now we use general Lov\'asz local lemma to show that $S$ contains no bad interval with positive probability. First we'll show the following lemma.
\begin{claim*}
The badness of interval $I=S[i,j]$ is mutually independent of the badness of all intervals that do not intersect with $I$.
\end{claim*}

\begin{proof}
Suppose the intervals before $I$ that do not intersect with $I$ are $I_1,\dots,I_m$, and those after $I$ are $I_1',\dots,I_{m'}'$. We denote the indicator variables of each interval being bad as $b$, $b_k$ and $b_{k'}'$. That is,
\[
b=
\begin{cases}
0 &\text{if $I$ is not bad}\\
1 &\text{if $I$ is bad}
\end{cases}
,\quad b_k=
\begin{cases}
0 &\text{if $I_k$ is not bad}\\
1 &\text{if $I_k$ is bad}
\end{cases}
,\quad b_{k'}'=
\begin{cases}
0 &\text{if $I_{k'}'$ is not bad}\\
1 &\text{if $I_{k'}'$ is bad}
\end{cases}
\]

First we prove that there exists $p\in (0,1)$ such that $\forall x_1,x_2,\dots,x_m\in\{0,1\}$,\[\Pr[b=1|b_k=x_k, k=1,\dots,m]=p\]

According to our construction, we can see that for any fixed prefix $S[1,i-1]$, the probability that $I$ is bad is a fixed real number $p'$. That is, \[\forall \text{ valid }\tilde{S}\in\Sigma^{i-1},\Pr[b=1|S[1,i-1]=\tilde{S}]=p'\]
This comes from the fact that, the sampling of the symbols in $S[i, k]$ only depends on the previous $h=min\{i-1, t-1\}$ different symbols, and up to a relabeling these $h$ symbols are the same $h$ symbols (e.g., we can relabel them as $\{1, \cdots, h\}$ and the rest of the symbols as $\{h+1, \cdots, |\Sigma|\}$). On the other hand the probability that $b=1$ remains unchanged under any relabeling of the symbols, since if two sampled symbols are the same, they will stay the same; while if they are different, they will still be different. Thus we have:
\begin{align*}
&\Pr[b=1|b_k=x_k, i=1,\dots,m]\\
=&\dfrac{\Pr[b=1,b_k=x_k, i=1,\dots,m]}{\Pr[b_k=x_k, k=1,\dots,m]}\\
=&\dfrac{\sum_{\tilde{S}}\Pr[b=1,S[1,i-1]=\tilde{S}]}{\sum_{\tilde{S}}\Pr[S[1,i-1]=\tilde{S}]}\\
=&\sum_{\tilde{S}}(\dfrac{\Pr[b=1,S[1,i-1]=\tilde{S}]}{\Pr[S[1,i-1]=\tilde{S}]}\dfrac{\Pr[S[1,i-1]=\tilde{S}]}{\sum_{\tilde{S}'}\Pr[S[1,i-1]=\tilde{S}']})\\
=&\sum_{\tilde{S}}(\Pr[b=1|S[1,i-1]=\tilde{S}]\dfrac{\Pr[S[1,i-1]=\tilde{S}]}{\sum_{\tilde{S}'}\Pr[S[1,i-1]=\tilde{S}']})\\
=&p'\sum_{\tilde{S}}\dfrac{\Pr[S[1,i-1]=\tilde{S}]}{\sum_{\tilde{S}'}\Pr[S[1,i-1]=\tilde{S}']}\\
=&p'
\end{align*}
In the equations, $\tilde{S}$ indicates all valid string that prefix $S[1,i-1]$ can be such that $b_k=x_k, k=1,\dots,m$. Hence, $b$ is independent of $\{b_k,k=1,\dots,m\}$.
Similarly, we can prove that the joint distribution of $\{b_{k'}',k'=1,\dots,m'\}$ is independent of that of $\{b,b_k,k=1,\dots,m\}$. Hence $b$ is independent of $\{b_k,b_{k'}',k=1,\dots,m,k'=1,\dots,m'\}$, which means, the badness of interval $I$ is mutually independent of the badness of all intervals that do not intersect with $I$.
\end{proof}

Obviously, an interval of length $l$ intersects at most $l+l'$ intervals of length $l'$. To use Lov\'asz local lemma, we need to find a sequence of real numbers $x_{i,k}\in[0.1)$ for intervals $S[i,k]$ for which
\[\Pr[S[i,k]\text{is bad}]\leq x_{i,k}\prod_{S[i,k]\cap S[i',k']\neq\emptyset}(1-x_{i',k'})\]

The rest of the proof is the same as that of Theorem 5.7 in \cite{haeupler2017synchronization}.

We propose $x_{i,k}=D^{-\varepsilon (k-i)}$ for some constant $D\geq 1$. Hence we only need to find a constant $D$ such that for all $S[i,k]$,
\[C^{-\varepsilon(k-i)}\leq D^{-\varepsilon (k-i)}\prod_{l=t}^n[1-D^{-\varepsilon l}]^{l+(k-i)}\]
That is, for all $l'\in\{1,...,n\}$,
\[C^{-l'}\leq D^{-l'}\prod_{l=t}^n[1-D^{-\varepsilon l}]^{\frac{l+l'}{\varepsilon}}\]
which means that
\[C\geq \dfrac{D}{\prod_{l=t}^n[1-D^{-\varepsilon l}]^{\frac{l/l'+1}{\varepsilon}}}\]
Notice that the righthand side is maximized when $n=\infty, l'=1$. Hence it's sufficient to show that
\[C\geq \dfrac{D}{\prod_{l=t}^{\infty}[1-D^{-\varepsilon l}]^{\frac{l+1}{\varepsilon}}}\]
Let $L=\max_{D>1}\dfrac{D}{\prod_{l=t}^{\infty}[1-D^{-\varepsilon l}]^{\frac{l+1}{\varepsilon}}}$. We only need to guarantee that $C>L$.

We claim that $L=\Theta(1)$. Since that $t=c_2\varepsilon^{-2}=\omega(\frac{\log\frac{1}{\varepsilon}}{\varepsilon})$,
\begin{align}
\dfrac{D}{\prod_{l=t}^{\infty}[1-D^{-\varepsilon l}]^{\frac{l+1}{\varepsilon}}}&<\dfrac{D}{\prod_{l=t}^{\infty}[1-\frac{l+1}{\varepsilon}D^{-\varepsilon l}]}\\
&<\dfrac{D}{1-\sum_{l=t}^{\infty}\frac{l+1}{\varepsilon}D^{-\varepsilon l}}\\
&=\dfrac{D}{1-\frac{1}{\varepsilon}\sum_{l=t}^{\infty}(l+1)D^{-\varepsilon l}}\\
&=\dfrac{D}{1-\frac{1}{\varepsilon}\frac{2tD^{-\varepsilon t}}{(1-D^{-\varepsilon})^2}}\\
&=\dfrac{D}{1-\frac{2}{\varepsilon^3}\frac{D^{-\frac{1}{\varepsilon}}}{(1-D^{-\varepsilon})^2}}
\end{align}

Inequality $(1)$ comes from the fact that $(1-x)^{\alpha}>1-\alpha x$, $(2)$ comes from he fact that $\prod_{i=1}^{\infty}(1-x_i)\geq 1-\sum_{i=1}^{\infty}x_i$ and $(3)$ is a result from $\sum_{l=t}^{\infty}(l+1)x^l=\frac{x^t(1+t-tx)}{(1-x)^2}<\frac{2tx^t}{(1-x)^2},x<1$.

We can see that for $D=7$, $\max_{\varepsilon}\{\frac{2}{\varepsilon^3}\frac{D^{-\frac{1}{\varepsilon}}}{(1-D^{-\varepsilon})^2}\}<0.9$. Therefore (5) is bounded by a constant, which means $L=\Theta(1)$ and the proof is complete.
\end{proof}
}
\fullOnly{\DetailedProofOfThmSyncStr}

Using a modification of an argument in \cite{haeupler2017synchronization}, we can also obtain a randomized construction. %as follows.

\begin{lemma}
\label{lem:randomalgo}
There exists a randomized algorithm which for any $\varepsilon\in(0,1)$ and any $ n\in \mathbb{N} $, constructs an $\varepsilon$-synchronization string of length $n$ over alphabet of size $O(\varepsilon^{-2})$ in expected time $O(n^5\log n)$.
\end{lemma}

\global\def\DetailedProofOfLemRandomAlgo{ 	% Define detailed proof of theorem
\shortOnly{\begin{proof}[Proof of Lemma~\ref{lem:randomalgo}]}
\fullOnly{\begin{proof}}

The algorithm is similar to that of Lemma 5.8 in \cite{haeupler2017synchronization}, using algorithmic Lov\'asz Local lemma \cite{moser2010constructive} and the extension in \cite{haeupler2011new}. It starts with a string sampled according to the sampling algorithm in the proof of Theorem \ref{syncStr}, over alphabet $\Sigma$ of size $C\varepsilon^{-2}$ for some large enough constant $C$.\ Then the algorithm checks all $O(n^2)$ intervals for a violation of the requirements for $\varepsilon$-synchronization string.\ If a bad interval is found, this interval is re-sampled by randomly choosing every symbol s.t. each one of them is different from the previous $t-1$ symbols, where $t = c'\varepsilon^{-2}$ with $c'$ being a constant smaller than $C$.

One subtle point of our algorithm is the following. Note that in order to apply the algorithmic framework of \cite{moser2010constructive} and \cite{haeupler2011new}, one needs the probability space to be sampled from $n$ independent random variables ${\cal P}=\{P_1, \cdots, P_n\}$ so that each event in the collection ${\cal A}=\{A_1, \cdots, A_m\}$ is determined by some subset of $\cal P$. Then, when some bad event $A_i$ happens, one only resamples the random variables that decide $A_i$. Upon first look, it may appear that in our application of the Lov\'asz Local lemma, the sampling of the $i$'th symbol depends on the the previous $h=min\{i-1, t-1\}$ symbols, which again depend on previous symbols, and so on. Thus the sampling of the $i$'th symbol depends on the sampling of all previous symbols.\ However, we can implement our sampling process as follows: for the $i$'th symbol we first independently generate a random variable $P_i$ which is uniform over $\{1, 2, \cdots, |\Sigma|-h\}$, then we use the random variables $\{P_1, \cdots, P_n\}$ to decide the symbols, in the following way. Initially we fix some arbitrary order of the symbols in $\Sigma$, then for $i=1, \cdots, n$, to get the $i$'th symbol, we first reorder the symbols $\Sigma$ so that the previous $h$ chosen symbols are labeled as the first $h$ symbols in $\Sigma$, and the rest of the symbols are ordered in the current order as the last $|\Sigma|-h$ symbols. We then choose the $i$'th symbol as the $(h+P_i)$'th symbol in this new order. In this way, the random variables $\{P_1, \cdots, P_n\}$ are indeed independent, and the $i$'th symbol is indeed chosen uniformly from the $|\Sigma|-h$ symbols excluding the previous $h$ symbols. Furthermore, the event of any interval $S[i, k]$ being bad only depends on the random variables $(P_i, \cdots, P_k)$ since no matter what the previous $h$ symbols are, they are relabeled as $\{1, \cdots, h\}$ and the rest of the symbols are labeled as $\{h+1, \cdots, |\Sigma|\}$. From here, the same sequence of $(P_i, \cdots, P_k)$ will result in the same behavior of $S[i, k]$ in terms of which symbols are the same. We can thus apply the same algorithm as in \cite{haeupler2017synchronization}.

Note that the time to get the $i$'th symbol from the random variables $\{P_1, \cdots, P_n\}$ is $O(n \log \frac{1}{\varepsilon})$ since we need $O(n)$ operations each on a symbol of size $C\varepsilon^{-2}$. Thus resampling each interval takes $O(n^2 \log \frac{1}{\varepsilon})$ time since we need to resample at most $n$ symbols.  For every interval, the edit distance can be computed using the Wagner-Fischer dynamic programming within $O(n^2 \log \frac{1}{\varepsilon})$ time. \cite{haeupler2011new} shows that the expected number of re-sampling is $O(n)$. The algorithm will repeat until no bad interval can be found.  Hence the overall expected running time is $O(n^5 \log \frac{1}{\varepsilon})$.

Note that without loss of generality we can assume that $\varepsilon>
1/\sqrt{n}$ because for smaller errors we can always use the indices
directly, which have alphabet size $n$. So the overall expected running time is $O(n^5 \log n)$.
\end{proof}
}
\fullOnly{\DetailedProofOfLemRandomAlgo}

%\subsection{Synchronization circle}
We can now construct an $\varepsilon$-synchronization circle using Theorem \ref{syncStr}. 
\begin{theorem}
\label{syncCircle}
For every $  \varepsilon\in(0,1)$ and every $n \in \mathbb{N}$, there exists an $\varepsilon$-synchronization circle $S$ of length $n$ over alphabet $\Sigma$ of size $O(\varepsilon^{-2})$.
\end{theorem}
%
%\begin{proof}[Proof Sketch]
%
%The main idea of the construction is to first construct two synchronization strings over two distinct alphabets, then concatenate the two strings ``head to head" and ``tail to tail". The key observation  is that, since the two alphabets are distinct, intervals from different strings have no common subsequence and intervals crossing the two strings can be reduced to intervals within one string in a case by case analysis. %considered by dividing into several easy cases.
%\end{proof}

\global\def\DetailedProofOfThmSyncCircle{ 	% Define detailed proof of theorem
\shortOnly{\begin{proof}[Proof of Theorem~\ref{syncCircle}]}
\fullOnly{\begin{proof}}
First, by Theorem \ref{syncStr}, we can have two $\varepsilon$-synchronization strings: $S_1$ with length $\lceil \frac{n}{2}\rceil$ over $\Sigma_1$ and $S_2$ with length $\lfloor\frac{n}{2}\rfloor$ over $\Sigma_2$. Let $\Sigma_1\cap\Sigma_2=\emptyset$ and $|\Sigma_1|=|\Sigma_2|= O(\varepsilon^{-2})$. Let $S$ be the concatenation of $S_1$ and $S_2$. Then $S$ is over alphabet $\Sigma=\Sigma_1\cup\Sigma_2$ whose size is $O(\varepsilon^{-2})$. Now we prove that $S$ is an $\varepsilon$-synchronization circle.

$\forall 1\leq m\leq n$, consider  string $S'=s_m,s_{m+1},\dots,s_n,s_1,s_2,\dots,s_{m-1}$. Notice that for two strings $T$ and $T'$ over alphabet $\Sigma$, $LCS(T,T')\leq\frac{\varepsilon}{2}(|T|+|T'|)$ is equivalent to $ED(T,T')\geq(1-\varepsilon)(|T|+|T'|)$. For any $i<j<k$, we call an interval $S'[i,k]$ good if $LCS(S'[i,j], S'[j+1,k])\leq \frac{\varepsilon}{2}(k-i)$. It suffices to show that $\forall 1\leq i,k\leq n$, the interval $S'[i,k]$ is good.

%If $m\geq\lceil\frac{n}{2}\rceil$, then intervals which are substrings of $S'[n-m+2,n-m+1+\lfloor\frac{n}{2}\rfloor]=S[1,\lfloor\frac{n}{2}\rfloor], S'[1,n-m+1]=S[m,n]$ and $S'[n-m+1+\lceil\frac{n}{2}\rceil,n]=S[\lceil\frac{n}{2}\rceil,m-1]$ are good intervals as they are also substring of $S_2$ and $S_1$.

Without loss of generality let's assume $ m \in [\lceil\frac{n}{2}\rceil, n ]$.

Intervals which are substrings of $S_1$ or $S_2$ are good intervals, since $S_1$ and $ S_2$ are $\varepsilon$-synchronization strings.

We are left with  intervals crossing the ends of $S_1$ or $S_2$.

\textbf{If $S'[i,k]$ contains $s_n,s_1$ but doesn't contain $s_{\lceil\frac{n}{2}\rceil}$:} If $j< n-m+1$, then there's no common subsequence between $s'[i,j]$ and $S'[n-m+2,k]$. Thus \[LCS(S'[i,j],S'[j+1,k])\leq LCS(S'[i,j],S'[j+1,n-m+1])\leq \frac{\varepsilon}{2}(n-m+1-i)<\frac{\varepsilon}{2}(k-i)\]
If $j\geq n-m+1$, then there's no common subsequence between $S'[j+1,k]$ and $S'[i,n-m+1]$. Thus
\[LCS(S'[i,j],S'[j+1,k])\leq LCS(S'[n-m+2,j],S'[j+1,k])\leq\frac{\varepsilon}{2}(k-(n-m+2))<\frac{\varepsilon}{2}(k-i)\]
Thus intervals of this kind are good.

\begin{figure}[H]
  \centering
  \includegraphics[width=8cm]{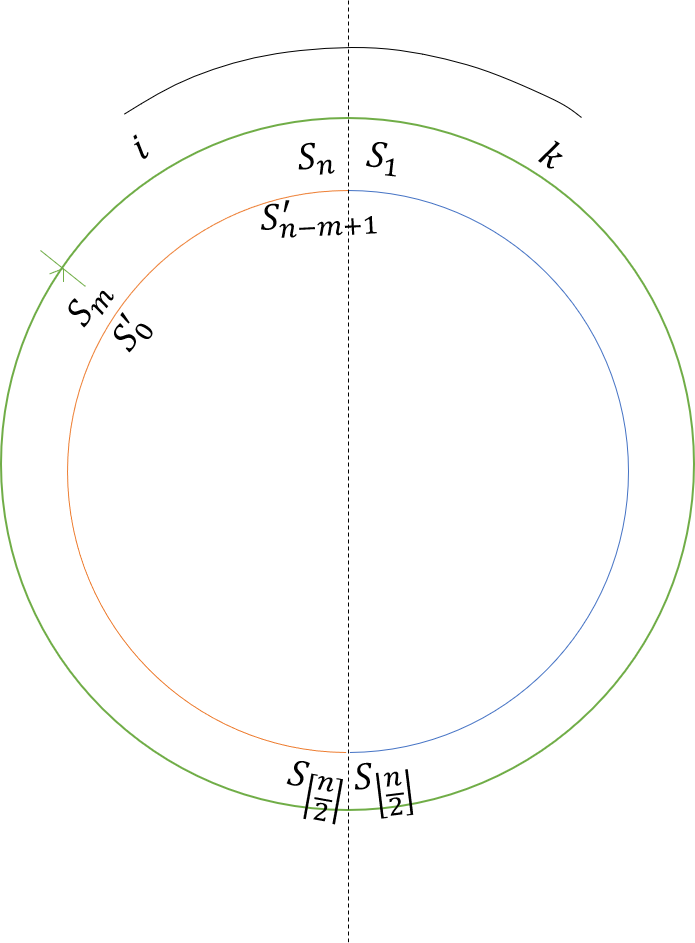}\\
  \caption{Example where $S'[i,k]$ contains $s_n,s_1$ but doesn't contain $s_{\lceil\frac{n}{2}\rceil}$}\label{sc1}
\end{figure}

\textbf{If $S'[i,k]$ contains $s_{\lfloor\frac{n}{2}\rfloor},s_{\lceil\frac{n}{2}\rceil}$ but doesn't contain $s_n$:} If $j\leq n-m+\lfloor\frac{n}{2}\rfloor+1$, then there's no common subsequence between $S'[i,j]$ and $S'[n-m+\lceil\frac{n}{2}\rceil+1,k]$, thus
\[LCS(S'[i,j],S'[j+1,k])\leq LCS(S'[i,j],S'[j+1,n-m+\lfloor\frac{n}{2}\rfloor+1])<\frac{\varepsilon}{2}(k-i)\]
If $j\geq n-m+\lfloor\frac{n}{2}\rfloor+1$, then there's no common subsequence between $S'[j+1,k]$ and $S'[i,n-m+\lfloor\frac{n}{2}\rfloor+1]$. Thus
\[LCS(S'[i,j],S'[j+1,k])\leq LCS(S'[n-m+\lceil\frac{n}{2}\rceil+1,j],S'[j+1,k])<\frac{\varepsilon}{2}(k-i)\]
Thus intervals of this kind are good.

\begin{figure}[H]
  \centering
  \includegraphics[width=8cm]{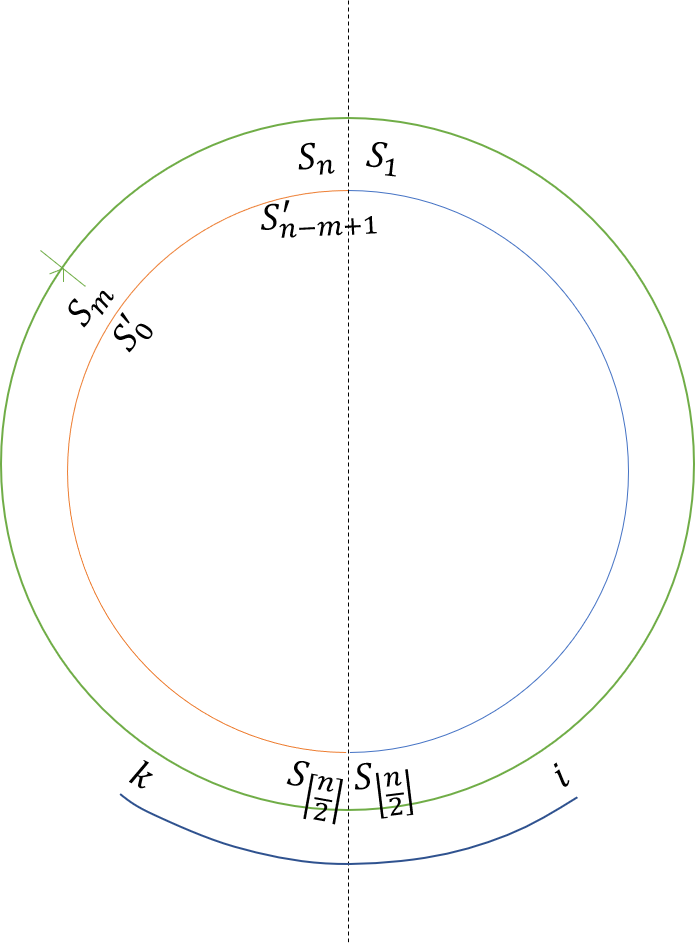}\\
  \caption{Example where $S'[i,k]$ contains $s_{\lfloor\frac{n}{2}\rfloor},s_{\lceil\frac{n}{2}\rceil}$}\label{sc2}
\end{figure}

\textbf{If $S'[i,k]$ contains $s_{\lceil\frac{n}{2}\rceil}$ and $s_n$:} If $n-m+2\leq j\leq n-m+\lfloor\frac{n}{2}\rfloor+1$, then the common subsequence is either that of $S'[i,n-m+1]$ and $S'[n-m+\lceil\frac{n}{2}\rceil+1,k]$ or that of $S'[n-m+2,j]$ and $S'[j+1,n-m+\lfloor\frac{n}{2}\rfloor+1]$.
%This comes from the fact that if a common subsequence is a match between $S'[i,n-m+1]$ and $S'[n-m+\lceil\frac{n}{2}\rceil+1,k]$, then this common sequence contains no match between $S'[n-m+2,j]$ and $S'[j+1,n-m+\lfloor\frac{n}{2}\rfloor+1]$.
This is because $\Sigma_1\cap\Sigma_2=\emptyset$.
Thus
\begin{align*}
&LCS(S'[i,j],S'[j+1,k])\\
\leq&\max\{LCS(S'[i,n-m+1],S'[n-m+\lceil\frac{n}{2}\rceil+1,k]),\\
&\qquad LCS(S'[n-m+2,j],S'[j+1,n-m+\lfloor\frac{n}{2}\rfloor+1])\}\\
<&\frac{\varepsilon}{2}(k-i)
\end{align*}
\begin{figure}[H]
  \centering
  \includegraphics[width=8cm]{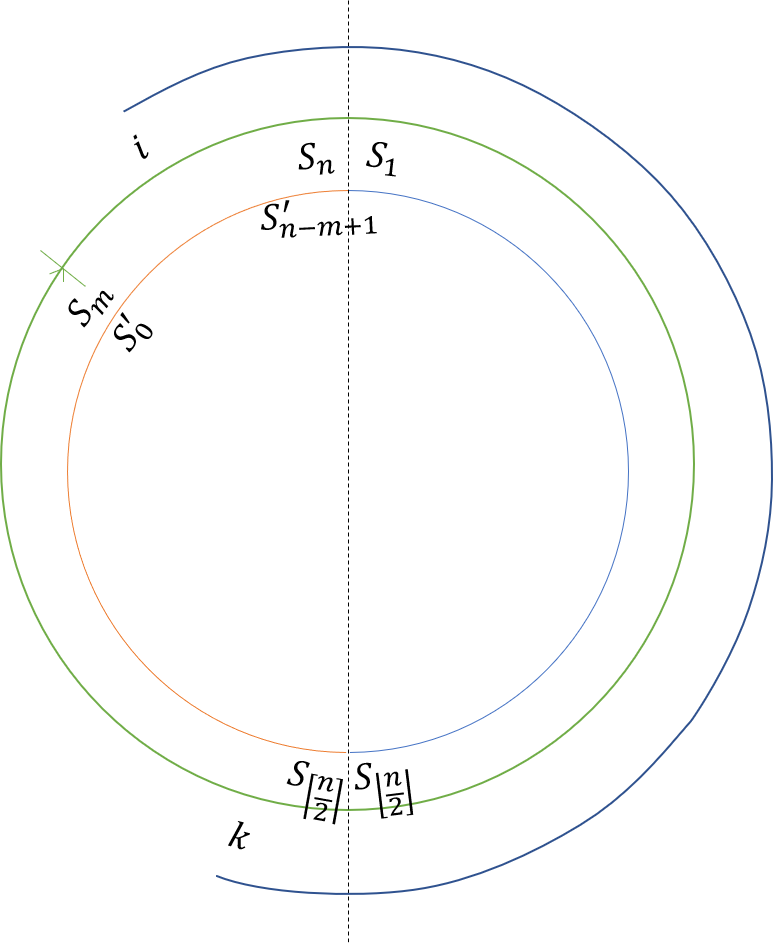}\\
  \caption{Example where $S'[i,k]$ contains $s_{\lceil\frac{n}{2}\rceil}$ and $s_n$}\label{sc3}
\end{figure}
If $j \leq n-m+1$, then there's no common subsequence between $S'[i,j]$ and $S'[n-m+2,n-m+\lfloor\frac{n}{2}\rfloor+1]$. Thus
\begin{align*}
&LCS(S'[i,j],S'[j+1,k])\\
\leq &LCS(S'[i,j],S'[j+1,n-m+1])+LCS( S'[i,j], S'[n-m+\lceil\frac{n}{2}\rceil+1,k])\\
< &    \frac{\varepsilon}{2}( n-m+1 - i  ) + \frac{\varepsilon}{2}( n - \lceil \frac{n}{2} \rceil )      \\
\leq &   \frac{\varepsilon}{2}(  n-m+1 - i  ) + \frac{\varepsilon}{2}(k- (n-m+2))      \\
= &   \frac{\varepsilon}{2}(k-1-i) \\
< &   \frac{\varepsilon}{2}(k-i)
\end{align*}
If $j\geq S'[n-m+\lceil\frac{n}{2}\rceil+1]$, the proof is similar to the case where $j\leq n-m+1$.

%Similarly we know that when $m\leq n-m+\lfloor\frac{n}{2}\rfloor+1$, $S'$ is still an $\varepsilon$-synchronization string.

This shows that $S'$ is an $\epsilon$-synchronization string. Thus by the definition of synchronization circle, the construction gives an  $\epsilon$-synchronization circle.
\end{proof}
}
\fullOnly{\DetailedProofOfThmSyncCircle}

%% file: clong.tex
\section{Deterministic Constructions of Long-Distance Synchronization Strings}\label{sec:LongDistSync}
In this section, we give deterministic constructions of synchronization strings. In fact, we consider a generalized version of synchronization strings, i.e., $f(l)$-distance $\varepsilon$-synchronization strings first defined by Haeupler and Shahrasbi~\cite{HS17c}. Throughout this section, $\eps$ is considered to be a constant in $(0, 1)$.

\begin{definition}[$f(l)$-distance $\varepsilon$-synchronization string]\label{def:distSynchStr}
A string $S \in \Sigma^n$ is an $f(l)$-distance $\varepsilon$-synchronization string if for every $1 \leq i<j\leq i'<j' \leq n + 1$,  $ED\left(S[i,j),S[i',j')\right) > (1-\varepsilon) (l)$ for $i'-j \leq f(l)$ where $l=j+j'-i-i'$.
\end{definition}

As a special case, $0$-distance synchronization strings are standard synchronization strings. %defined by Haeupler and Shahrasbi~\cite{haeupler2017synchronization}.

Similar to \cite{HS17c}, we focus on $f(l) = n \cdot \mathbbm{1}_{l>c\log n}$ where $\mathbbm{1}_{l>c\log n}$ is the indicator function for $l>c\log n$. This function considers the edit distance of all pairs of large intervals and adjacent small intervals. %We refer to such synchronization string as \emph{$c$-long-distance $\varepsilon$-synchronization strings}.

\begin{definition}[$c$-long-distance $\varepsilon$-synchronization strings]
We call $n \cdot \mathbbm{1}_{l>c\log n}$-distance $\varepsilon$-synchronization strings \emph{$c$-long-distance $\varepsilon$-synchronization strings}.
\end{definition}

\subsection{Polynomial Time Constructions of Long-Distance Synchronization Strings}\label{sec:LLLConstruction}

Here, by combining the deterministic Lov\'{a}sz local lemma of Chandrasekaran et al.\ \cite{chandrasekaran2013deterministic} and the non-uniform sample space used in Theorem \ref{syncStr}, we give a deterministic polynomial-time construction of $c$-long $\varepsilon$-synchronization strings over an alphabet of size $O(\varepsilon^{-2})$.

First we recall the following property of $c-$long synchronization strings.
%
%\begin{lemma}[Lemma 4.3 in \cite{HS17c}]
%\label{lem:longdistancereduction}
%If $S$ is a string and there are two intervals $i_1<j_1\leq i_2 < j_2$ of total length $l = j_1 - i_1 + j_2 - i_2$ such that $ED(S[i_1,j_1), S[i_2,j_2)) \le (1-\varepsilon) l$, then there also exist intervals $i_1 \leq i'_1<j'_1\leq i'_2 < j'_2 \leq i_2$ of total length $l' \in \{\lceil l/2\rceil-1,\lceil l/2\rceil, \lceil l/2\rceil+1\}$ such that $ED(S[i'_1,j'_1),S[i'_2,j'_2)) \leq (1-\varepsilon) l'$.
%\end{lemma}
%
%This lemma implies that, for pairs of non-adjacent intervals, we only need to care about those of total length between $[c\log n, 2c\log n]$.

\begin{lemma}[Corollary 4.4 of \cite{HS17c}]
\label{cor:shortDistanceSufficient}
If $S$ is a string which satisfies the $c$-long-distance $\varepsilon$-synchronization property for any two non-adjacent intervals of total length $2c\log n$ or less, then it satisfies the property for all pairs of non-adjacent intervals.
\end{lemma}
%\begin{corollary}
%If $S$ is a string which satisfies the $c^l$-distance $\varepsilon$-synchronization property for any two intervals of total length between $\log_c n$ and $2\log_c n$ then it satisfies the property for all pairs of intervals of total length $\log_c n$ or longer.
%\end{corollary}

%Next we give our theorem for a deterministic polynomial time construction of long-distance synchronization strings. The main idea is the following.
We now have the following theorem.
\begin{theorem}
\label{thm:polyConstruction}
For any $n\in \mathbb{N}$ and
any constant $\varepsilon\in (0,1)$,
there is a deterministic construction of a $O(1/\varepsilon)$-long-distance $\varepsilon$-synchronization string of length $n$, over an alphabet of size $O(\varepsilon^{-2})$, in time $\poly(n)$.
\end{theorem}

\global\def\DetailedProofOfThmPolyConstruction{ 	% Define detailed proof of theorem
\shortOnly{\begin{proof}[Proof of Theorem~\ref{thm:polyConstruction}]}
\fullOnly{\begin{proof}}
To prove this, we will use the Lov\a'sz Local Lemma and its deterministic algorithm in \cite{chandrasekaran2013deterministic}.
Suppose the alphabet is $\Sigma$ with $|\Sigma|=q = c_1\varepsilon^{-2}$ where $c_1$ is a constant. Let $t=c_2\varepsilon^{-2}$ and $0< c_2 < c_1$. We denote $|\Sigma|-t$ as $q$. The sampling algorithm of string $S$ ($1$-index based)is as follows:
\begin{itemize}
  \item Initialize an arbitrary order for $\Sigma$.
  \item For $i$th symbol:
    \begin{itemize}
    \item Denote $h=\min\{t-1, i-1\}$. Generate a random variable $P_i$ uniformly over $\{1,2,\dots,|\Sigma|-h\}$.
    \item Reorder $\Sigma$ such that the previous $h$ chosen symbols are labeled as the first $h$ symbols in $\Sigma$, and the rest are ordered in the current order as the last $|\Sigma|-h$ symbols.
      \item Choose the $(P_i+h)$'th symbol in this new order as $S[i]$.
    \end{itemize}
\end{itemize}
Define the bad event $A_{i_1,l_1,i_2,l_2}$ as intervals $S[i_1,i_1+l_1)$ and $S[i_2,i_2+l_2)$ violating the $c = O(1/\varepsilon)$-long-distance synchronization string property for $i_1+l_1\leq i_2$. In other words, $A_{i_1,l_1,i_2,l_2}$ occurs if and only if $ED(S[i_1,i_1+l_1),S[i_2,i_2+l_2))\leq(1-\varepsilon)(l_1+l_2)$, which is equivalent to $LCS(S[i_1,i_1+l_1),S[i_2,i_2+l_2))\geq \frac{\varepsilon}{2}(l_1+l_2)$.

Note that according to the definition of $c$-long distance $\varepsilon$-synchronization string and Lemma \ref{cor:shortDistanceSufficient}, we only need to consider $A_{i_1,l_1,i_2,l_2}$ where $l_1+l_2<c\log n$ and $c\log n\leq l_1+l_2\leq 2c\log n$. Thus we can upper bound the probability of $A_{i_1,l_1,i_2,l_2}$,
\begin{eqnarray*}
\Pr\left [A_{i_1, l_1, i_2, l_2}\right ] &\le& {l_1 \choose \varepsilon(l_1+l_2)/2}{l_2 \choose \varepsilon(l_1+l_2)/2}({\left|\Sigma\right|-t})^{-\frac{\varepsilon (l_1+l_2)}{2}}\\
&\le& \left(\frac{l_1 e}{\varepsilon(l_1+l_2)/2}\right)^{\varepsilon(l_1+l_2)/2} \left(\frac{l_2 e}{\varepsilon(l_1+l_2)/2}\right)^{\varepsilon(l_1+l_2)/2}({\left|\Sigma\right|-t})^{-\frac{\varepsilon (l_1+l_2)}{2}}\\
&=&\left(\frac{2e\sqrt{l_1l_2}}{\varepsilon(l_1+l_2)\sqrt{|\Sigma|-t}}\right)^{\varepsilon(l_1+l_2)}\\
&\le&  \left(\frac{el}{\varepsilon l \sqrt{|\Sigma|-t}}\right)^{\varepsilon l} = \left(\frac{e}{\varepsilon \sqrt{|\Sigma|-t}}\right)^{\varepsilon l}=\hat{C}^{\varepsilon l},
\end{eqnarray*}
where $l=l_1+l_2$ and $\hat{C}$ is a constant which depends on $c_1$ and $c_2$.

However, to apply the deterministic Lov\a'sz Local Lemma (LLL), we need to have two additional requirements.
The first  requirement is that each bad event depends on up
to logarithmically many variables, and the second is that the inequalities in the Lov\a'sz Local Lemma hold with a constant exponential
slack \cite{HS17c}.

The first requirement may not be true under the current definition of badness.\ Consider for example the random variables $P_{i_1},\dots,P_{i_1+l_1-1}, P_{i_2}, P_{i_2+l_2-1}$ for a pair of split intervals $S[i_1,i_1+l_1), S[i_2,i_2+l_2)$ where the total length $l_1+l_2$ is at least $2c\log n$. The event $A_{i_1,l_1,i_2,l_2}$ may depend on too many random variables (i.e., $P_{i_1}, \ldots, P_{i_2+l_2-1}$).

To overcome this, we redefine the badness of the split interval $S[i_1,i_1+l_1)$ and $S[i_2,i_2+l_2)$ as follows: let $B_{i_1,l_1,i_2,l_2}$ be the event that there exists $P_{i_1+l_1},\dots,P_{i_2-1}$ (i.e., the random variables chosen between the two intervals) such that the two intervals generated by $P_{i_1}\dots,P_{i_1+l_1-1}$ and $P_{i_2},\dots,P_{i_2+l_2-1}$ (together with  $P_{i_1+l_1},\dots,P_{i_2-1}$) makes $LCS(S[i_1,i_1+l_1),S[i_2,i_2+l_2))\geq \frac{\varepsilon}{2}(l_1+l_2)$ according to the sampling algorithm. Note that if $B_{i_1,l_1,i_2,l_2}$ does not happen, then certainly $A_{i_1,l_1,i_2,l_2}$ does not happen.

Notice that with this new definition of badness, $B_{i_1,l_1,i_2,l_2}$ is independent of $\{P_{i_1+l_1}, \ldots, P_{i_2-1}\}$ and only depends on $\{P_{i_1}\dots,P_{i_1+l_1-1}, P_{i_2}, \ldots, P_{i_2+l_2}\}$. In particular, this implies that $B_{i_1,l_1,i_2,l_2}$ is independent of the badness of all other intervals which have no intersection with $(S[i_1,i_1+l_1)$, $S[i_2,i_2+l_2))$.

We now bound $\Pr[B_{i_1,l_1,i_2,l_2}]$. When considering the two intervals $S[i_1,i_1+l_1)$, $S[i_2,i_2+l_2)$ and their edit distance under our sampling algorithm, without loss of generality we can assume that the order of the alphabet at the point of sampling $S[i_1]$ is $(1, 2, \ldots, q)$ just by renaming the symbols. Now, if we fix the order of the alphabet at the point of sampling $S[i_2]$ in our sampling algorithm, then $S[i_2, i_2+l_2)$ only depends on $\{P_{i_2}, \ldots, P_{i_2+l_2}\}$ and thus $LCS(S[i_1,i_1+l_1),S[i_2,i_2+l_2))$ only depends on $\{P_{i_1}\dots,P_{i_1+l_1-1}, P_{i_2}, \ldots, P_{i_2+l_2}\}$.

Conditioned on any fixed order of the alphabet at the point of sampling $S[i_2]$, we have that $LCS(S[i_1,i_1+l_1),S[i_2,i_2+l_2))\geq \frac{\varepsilon}{2}(l_1+l_2)$ happens with probability at most ${\hat{C}}^{\varepsilon l}$ by the same computation as we upper bound $\Pr[ A_{i_1, l_1, i_2, l_2} ]$. Note that there are at most $q!$ different orders of the alphabet. Thus by a union bound we have
\begin{align*}
 \Pr[B_{i_1, l_1, i_2, l_2}]
\leq\hat{C}^{\varepsilon l}\times q!
= C^{\varepsilon l},
\end{align*}
for some constant $C$.

In order to meet the second requirement of the deterministic algorithm of LLL, we also need to find real numbers $x_{i_1,i_1+l_1,i_2,i_2+l_2}\in[0,1]$ such that for any $B_{i_1,l_1,i_2,l_2}$,
\begin{equation*}\label{eqn:LLLcondition2}
\Pr[B_{i_1, l_1, i_2, l_2}] \le \left[x_{i_1, l_1, i_2, l_2} \prod_
{\left[S[i_1, i_1+l_1)\cup S[i_2, i_2+l_2)\right]\cap[S[i'_1, i'_1+l'_1)\cup S[i'_2, i'_2+l'_2)]\neq \emptyset}
 (1-x_{i'_1, l'_1, i'_2, l'_2})\right]^{1.01}.
\end{equation*}
%The constant exponential slack comes from the requirement of polynomial time deterministic algorithmic LLL.
We propose $x_{i_1, l_1, i_2, l_2} = D^{-\varepsilon (l_1 + l_2)}$ for some $D> 1$ to be determined later. $D$ has to be chosen such that for any $i_1, l_1, i_2, l_2$ and $l=l_1 + l_2$:

\begin{align}
\left(\frac{e}{\varepsilon  \sqrt{|\Sigma|}}\right)^{\varepsilon l} &\le& \left[D^{-\varepsilon l} \prod_
{[S[i_1, i_1+l_1)\cup S[i_2, i_2+l_2)]\cap[S[i'_1, i'_1+l'_1)\cup S[i'_2, i'_2+l'_2)]\neq \emptyset}
 \left(1-D^{-\varepsilon (l'_1 + l'_2)}\right)\right]^{1.01}
\end{align}

Notice that
\begin{eqnarray}
&&D^{-\varepsilon l} \prod_{[S[i_1, i_1+l_1)\cup S[i_2, i_2+l_2)]\cap[S[i'_1, i'_1+l'_1)\cup S[i'_2, i'_2+l'_2)]\neq \emptyset}
 \left(1-D^{-\varepsilon (l'_1 + l'_2)}\right)\\
&\ge& D^{-\varepsilon l} \prod_{l'=c\log n}^{2c\log n}\prod_{l'_1=1}^{l'}
\left(1-D^{-\varepsilon l'}\right)^{\left[(l_1+l'_1)+(l_1+l'_2)+(l_2+l'_1)+(l_2+l'_2)\right] n}
\nonumber\\
&&\times \prod_{l''=t}^{c\log n} \left(1-D^{-\varepsilon l''}\right)^{l+l''}\label{eqn:countPairOfIntervals}\\
 &=& D^{-\varepsilon l} \prod_{l'=c\log n}^{2c\log n}\prod_{l'_1=1}^{l'}
 \left(1-D^{-\varepsilon l'}\right)^{4(l+l') n}
 \times \prod_{l''=t}^{c\log n} \left(1-D^{-\varepsilon l''}\right)^{l+l''}
 \\&=&
 D^{-\varepsilon l} \prod_{l'=c\log n}^{2c\log n}
 \left(1-D^{-\varepsilon l'}\right)^{4l'(l+l') n}
  \times \left[\prod_{l''=t}^{c\log n} \left(1-D^{-\varepsilon l''}\right)\right]^{l}
  \times \prod_{l''=t}^{c\log n} \left(1-D^{-\varepsilon l''}\right)^{l''}\\
 &\ge&D^{-\varepsilon l} \left(1-\sum_{l'=c\log n}^{2c\log n}\left(4l'(l+l') n\right)D^{-\varepsilon l'}\right)
\nonumber\\&&
\times \left[1-\sum_{l''=t}^{c\log n} D^{-\varepsilon l''}\right]^{l}
 \times  \left(1-\sum_{l''=t}^{c\log n}l''D^{-\varepsilon l''}\right)
\label{eqn:linearization}\\
  &\ge&D^{-\varepsilon l} \left(1-\sum_{l'=c\log n}^{2c\log n}\left(4\cdot 2c\log n(2c\log n+2c\log n) n\right)D^{-\varepsilon l'}\right)\\
&&
\times \left[1-\sum_{l''=t}^{\infty} D^{-\varepsilon l''}\right]^{l}
 \times  \left(1-\sum_{l''=t}^{\infty}l''D^{-\varepsilon l''}\right)
\\
  &=&D^{-\varepsilon l} \left(1-\sum_{l'=c\log n}^{2c\log n}\left(32c^2n\log^2 n\right)D^{-\varepsilon l'}\right)
\times \left[1-\frac{D^{-c_2\varepsilon^{-1}}}{1-D^{-\varepsilon }}\right]^{l}
\nonumber\\&&
 \times  \left(1-\frac{D^{-c_2/\varepsilon}(D^{-\varepsilon}+c_2/\varepsilon^2-c_2D^{-\varepsilon}/\varepsilon^2)}{(1-D^{-\varepsilon})^2}\right)
 \\&\ge& D^{-\varepsilon l} \left(1-32c^3n\log^3 nD^{-\varepsilon c\log n}\right)
 \left[1-\frac{D^{-c_2\varepsilon^{-1}}}{1-D^{-\varepsilon }}\right]^{l}
\nonumber\\&&
\times \left(1-\frac{D^{-c_2/\varepsilon}(D^{-\varepsilon}+c_2/\varepsilon^2-c_2D^{-\varepsilon}/\varepsilon^2)}{(1-D^{-\varepsilon})^2}\right)\label{eqn:BadEventProbLowerBound}
\end{eqnarray}

Equation \ref{eqn:countPairOfIntervals} holds because there are two kinds of pairs of intervals. The first kind contains all pairs of intervals whose total length is between $c\log n$ and $2c\log n$ intersecting with $S[i_1,i_1+l_1)$ or $S[i_2,i_2+l_2)$. The number of such pairs is at most $(l_1+l'_1)+(l_1+l'_2)+(l_2+l'_1)+(l_2+l'_2)$. The second kind contains all adjacent intervals of total length less than $c\log n$. Notice that according to our sampling algorithm, every $t$ consecutive symbols are distinct, thus any adjacent intervals whose total length is less than $t$ cannot be bad. Hence the second term contains intervals such that $t\leq l''=l_1''+l_2''\leq c\log n$.

The rest of the proof is the same as that of Theorem 4.5 in \cite{HS17c}.

Equation \ref{eqn:linearization} comes from the fact that for $0<x,y<1$:
\[(1-x)(1-y)>1-x-y\]

For $D=2$ and $c = 2/\varepsilon$,
\[\lim_{\varepsilon\rightarrow 0}\frac{2^{-c_2/\varepsilon}}{1-2^{-\varepsilon}}=0\]
Thus, for sufficiently small $\varepsilon$, $\frac{2^{-c_2/\varepsilon}}{1-2^{-\varepsilon}}<\frac{1}{2}$. Moreover,
\[32c^2n\log^2 n D^{-\varepsilon l'}=\frac{2^8}{\varepsilon^3}\frac{\log^3 n}{n}=o(1)\]
Finally, for sufficiently small $\varepsilon$, $1-\frac{D^{-c_2/\varepsilon}(D^{-\varepsilon}+c_2/\varepsilon^2-c_2D^{-\varepsilon}/\varepsilon^2)}{(1-D^{-\varepsilon})^2}>2^{-\varepsilon}$. Therefore, for sufficiently small $\varepsilon$ and sufficiently large $n$, \ref{eqn:BadEventProbLowerBound} is satisfied under the condition:
\begin{eqnarray*}
&&D^{-\varepsilon l} \prod_{[S[i_1, i_1+l_1)\cup S[i_2, i_2+l_2)]\cap[S[i'_1, i'_1+l'_1)\cup S[i'_2, i'_2+l'_2)]\neq \emptyset}
 \left(1-D^{-\varepsilon (l'_1 + l'_2)}\right)\\
&\geq & 2^{-\varepsilon l}(1-\frac{1}{2})(2^{-\varepsilon})^l(1-\frac{1}{2})\geq \frac{4^{-\varepsilon l}}{4}
\end{eqnarray*}

So for LLL to work, the following should be guaranteed:
\[\left(\frac{e}{\varepsilon\sqrt{|\Sigma|-t}}\right)^{\frac{\varepsilon l}{1.01}}\leq\frac{4^{-\varepsilon l}}{4}\Leftarrow \frac{4^{2.02(1+\varepsilon)e^2}}{\varepsilon^2}\leq |\Sigma|-t\]
Hence the second requirement holds for $|\Sigma|-t=\frac{4^{4.04}e^2}{\varepsilon^2}=O(\varepsilon^{-2})$.
\end{proof}
}
\fullOnly{\DetailedProofOfThmPolyConstruction}

%Since a $c$-long synchronization string is a synchronization string, we immediately have the following.
\begin{corollary}
\label{corollary:polyepssyncstring}
For any $n\in \mathbb{N}$ and
%$\varepsilon=\Omega((\log n)^{-\frac{1}{2}})$,
any constant $\varepsilon\in (0,1)$,
%$c>1$,
there is a deterministic construction of an $\varepsilon$-synchronization string of length $n$, over an alphabet of size $O(\varepsilon^{-2})$, in time $\poly(n)$.
\end{corollary}

By a similar concatenation construction used in the proof of Theorem \ref{syncCircle}, we also have a deterministic construction for synchronization circles.
\begin{corollary}
\label{corollary:polyepssynccircle}
For any $n\in \mathbb{N}$ and
%$\varepsilon=\Omega((\log n)^{-\frac{1}{2}})$,
any constant $\varepsilon\in (0,1)$,
%$c>1$,
there is a deterministic construction of an $\varepsilon$-synchronization circle of length $n$, over an alphabet of size $O(\varepsilon^{-2})$, in time $\poly(n)$.
\end{corollary}

\subsection{Deterministic linear time constructions of $c$-long distance $\varepsilon$-synchronization string}

Here we give a much more efficient construction of a $c$-long distance $\varepsilon$-synchronization string, using synchronization circles and standard error correcting codes. We show that the following algorithm gives a construction of $c$-long distance synchronization strings.

\begin{algorithm}[H]
\caption{Explicit Linear Time Construction of $c$-long distance $\varepsilon$-synchronization string}
\label{algo:clonglineartime}
\begin{algorithmic}
\STATE \textbf{Input:}
\begin{itemize}
\item An ECC $\hat{\mathcal{C}}\subset \Sigma_{\hat{C}}^m$, with distance $\delta m$ and block length $m= c \log n $.
\item An $\varepsilon_0$-synchronization circle $SC=(sc_1,\dots,sc_m)$ of length $m$ over alphabet $\Sigma_{SC}$.
\end{itemize}
\STATE \textbf{Operations:}
\begin{itemize}
\item Construct a code $\mathcal{C}\subset \Sigma^m$ such that
\[\mathcal{C} = \{((\hat{c}_1,sc_1),\dots,(\hat{c}_m,sc_m))|(\hat{c}_1,\dots,\hat{c}_m)\in\hat{\mathcal{C}}\}\]
 where $\Sigma = \Sigma_{\hat{C}}\times\Sigma_{SC}$.
\item Let $S$ be concatenation of all codewords $\mathcal{C}_1,\dots,\mathcal{C}_N$ from $\mathcal{C}$.
\end{itemize}
\STATE \textbf{Output:} $S$.
\end{algorithmic}
\end{algorithm}

To prove the correctness, we first recall the following theorem from  \cite{haeupler2017synchronization}.
\begin{theorem}[Theorem 4.2 of \cite{haeupler2017synchronization}]
\label{insdelCode}
Given an $\varepsilon_0$-synchronization string $S$ with length $n$, 
%alphabet $\Sigma_S$, 
and an efficient ECC $\mathcal{C}$ with block length $n$, 
%alphabet $\Sigma_C$, 
that corrects up to $n\delta \frac{1+\varepsilon_0}{1-\varepsilon_0}$ half-errors, one can obtain an insertion/deletion code $\mathcal{C}'$  that can be decoded from up to $n\delta$ deletions, where $ \mathcal{C}' = \{(c'_1, \ldots, c'_n)| \forall  i\in [n],  c'_i = (c_i, S[i]), (c_1,\ldots, c_n) \in \mathcal{C}  \} $.
\end{theorem}

We have the following property of longest common subsequence.
\begin{lemma}
\label{algo:lcs}
Suppose $T_1$ is the concatenation of $\ell_1$ strings, $T_1=S_1\circ\dots\circ S_{\ell_1}$ and $T_2$ is the concatenation of $\ell_2$ strings, $T_2=S'_1\circ\dots\circ S'_{\ell_2}$. If  there exists an integer $t$ such that for all $i, j$, we have $LCS(S_i, S'_j) \leq t$, then we have $LCS(T_1, T_2) \leq (\ell_1+\ell_2)t$.
%Suppose $T_1$ and $T_2$ are concatenations of different codewords from $\mathcal{C}$, where all the codewords used are from $\mathcal{C}$ and different. Then $LCS(T_1,T_2)\leq (\ell_1+\ell_2)\alpha m$.
\end{lemma}
%The proof is deferred to the Appendix.

\global\def\DetailedProofOfalgolcsLemma{ 	% Define detailed proof of theorem
\shortOnly{\begin{proof}[Proof of Lemma~\ref{algo:lcs}]}
\fullOnly{\begin{proof}}

We 
%rename the strings in $T_1$ by $S_1, \cdots, S_{\ell_1}$ and 
rename the strings in $T_2$ by $S_{\ell_1+1}, \cdots, S_{\ell_1+\ell_2}$. Suppose the longest common subsequence between $T_1$ and $T_2$ is $\tilde{T}$, which can be viewed as a matching between $T_1$ and $T_2$.

we can divide $\tilde{T}$ sequentially into disjoint intervals, where each
interval corresponds to a common subsequence between a different pair
of strings $(S_i, S_j)$, where $S_i$ is from $T_1$ and $S_j$ is from $T_2$. In
addition, if we look at the intervals from left to right, then for any
two consecutive intervals and their corresponding pairs $(S_i, S_j)$ and $(S_{i'}, S_{j'})$, we must have $i' \geq i$ and $j' \geq j$  since the matchings which correspond to two intervals cannot cross each other. Furthermore either $i' > i$
or $j' > j$ as the pair $(S_i, S_j)$ is different from $(S_{i'}, S_{j'})$.

Thus, starting from the first interval, we can
label each interval with either $i$ or $j $ such that every interval
receives a different label, as follows. We label the first interval using either $i$ or $j$. Then, assuming we have already labeled some intervals and now look at the next interval. Without loss of generality assume that the previous interval is labeled using $i$, now if the current $i' > i$ then we can label the current interval using $i'$; otherwise we must have $j'>j$ so we can label the current interval using $j'$. Thus the total number of the labels is at
most $l_1+l_2$, which means the total number of the intervals is also at
most $l_1+l_2$. Note that each interval has length at most $t$, therefore we can upper bound $LCS(T_1, T_2)$ by
 $(l_1+l_2)t$.
\end{proof}
}
\fullOnly{\DetailedProofOfalgolcsLemma}

%Then we show that the algorithm outputs a synchronization circle.
\begin{lemma}
\label{mainlemmaofalgo}
The output $S$ in Algorithm \ref{algo:clonglineartime} is an $\varepsilon_1$-synchronization circle, where $\varepsilon_1 \leq 10(1-\frac{1-\varepsilon_0}{1+\varepsilon_0}\delta) $.
\end{lemma}

%\begin{proof}[Proof Sketch]
%
%Let $S[i, k]$ be an arbitrary substring of $S$. We mainly consider two cases.
%In the first case, where $ k-i >m $, we mainly use the fact that the edit distance between two codewords is large. In the second case, where $k-i \leq m$, we mainly use the property of  the synchronization circle. 
%%Detailed proof is deferred to the appendix.
%\end{proof}

\global\def\DetailedProofOfMainLemma{ 	% Define detailed proof of theorem
\shortOnly{\begin{proof}[Proof of Lemma~\ref{mainlemmaofalgo}]}
\fullOnly{\begin{proof}}
Suppose $\hat{\mathcal{C}}$ can correct up to $\delta m$ half-errors. Then according to lemma \ref{insdelCode}, $\mathcal{C}$ can correct up to $\frac{1-\varepsilon_0}{1+\varepsilon_0} \delta m$ deletions.

Let $\alpha=1-\frac{1-\varepsilon_0}{1+\varepsilon_0}\delta$. Notice that $\mathcal{C}$ has the following properties:
\begin{enumerate}
  \item $LCS(\mathcal{C}) = \max_{c_1, c_2\in C}{LCS(c_1, c_2)} \leq \alpha m$
  \item Each codeword in $\mathcal{C}$ is an $\varepsilon$-synchronization circle over $\Sigma$.
\end{enumerate}

Consider any shift of the start point of $S$, we only need to prove that $\forall 1\leq i < j < k\leq n, LCS(S[i,j],S[j+1,k])<\frac{\varepsilon_1}{2}(k-i)$.

Suppose $S_1=S[i,j]$ and $S_2=S[j+1,k]$. Let $\varepsilon_1 = 10\alpha$.

\textbf{Case 1:} $k-i>m$. Let $|S_1|=s_1$ and $|S_2|=s_2$, thus $s_1+s_2 > m$. If we look at each $S_h$ for $h=1, 2$, then $S_h$ can be divided into some consecutive codewords, plus at most two incomplete codewords at both ends. In this sense each $S_h$ is the concatenation of $\ell_h$ strings with $\ell_h < \frac{s_h}{m}+2$. An example of the worst case appears in Figure \ref{4}.

\begin{figure}[H]
  \centering
  \includegraphics[width=8cm]{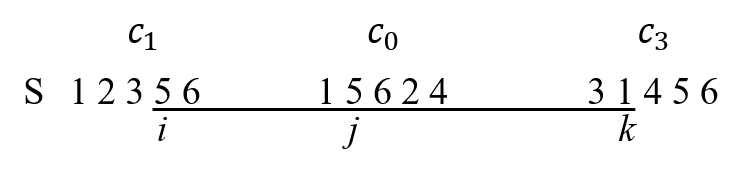}\\
  \caption{Example of the worst case, where $j$ splits a codeword, and there are two incomplete codewords at both ends.}\label{4}
\end{figure}

Now consider the longest common subsequence between any pair of these strings where one is from $S_1$ and the other is from $S_2$, we claim that the length of any such longest common subsequence is at most $\alpha m$. Indeed, if the pair of strings are from two different codewords, then by the property of the code $\mathcal{C}$ we know the length is at most $\alpha m$. On the other hand, if the pair of strings are from a single codeword (this happens when $j$ splits a codeword, or when $S[i]$ and $S[k]$ are in the same codeword), then they must be two disjoint intervals within a codeword. In this case, by the property that any codeword is also a synchronization circle, the length of the longest common subsequence of this pair is at most $\frac{\varepsilon_0}{2} m$.

Note that $\alpha=1-\frac{1-\varepsilon_0}{1+\varepsilon_0}\delta \geq  1-\frac{1-\varepsilon_0}{1+\varepsilon_0} =\frac{2\varepsilon_0}{1+\varepsilon_0} \geq \varepsilon_0$ (since $\delta, \varepsilon_0 \in (0, 1)$). Thus $\frac{\varepsilon_0}{2} m < \alpha m$. Therefore, by Lemma~\ref{algo:lcs}, we have

\begin{equation}
\begin{aligned}
& LCS(S_1, S_2) \\
< &(\frac{s_1}{m}+2+\frac{s_2}{m}+2) \alpha m \\
= &\alpha (s_1+s_2+4 m) \\
< & 5 \alpha (s_1+s_2) \\
= & 5 \alpha (k-i) = \frac{\varepsilon_1}{2}(k-i)
\end{aligned}
\end{equation}

\textbf{Case 2:} If $k-i\leq m$, then according to the property of synchronization circle $ SC $, we know that the longest common subsequence of $S_1$ and $S_2$ is less than $\frac{\varepsilon_0}{2}(k-i)\leq \alpha(k-i)\leq\frac{\varepsilon_1}{2}(k-i)$.

As a result, the longest common subsequence of $S[i,j]$ and $S[j+1,k]$ is less than $\frac{\varepsilon_1}{2}(k-i)$, which means that $S$ is an $\varepsilon_1$-synchronization circle.
\end{proof}
}
\fullOnly{\DetailedProofOfMainLemma}

Similarly, we also have the following lemma.
\begin{lemma}
\label{lem:construction}
The output $S$ of algorithm \ref{algo:clonglineartime} is a $c$-long distance $\varepsilon$-synchronization string of length $n=Nm$ where $N$ is the number of codewords in $\mathcal{C}$, $\varepsilon=12(1-\frac{1-\varepsilon_0}{1+\varepsilon_0}\delta)$.
\end{lemma}
%Detailed proof is deferred to the appendix.

\global\def\DetailedProofOfLemConstruction{ 	% Define detailed proof of theorem
\shortOnly{\begin{proof}[Proof of Lemma~\ref{lem:construction}]}
\fullOnly{\begin{proof}}

By Lemma \ref{mainlemmaofalgo},
$S$ is an $\varepsilon_1$-synchronization string, thus the length of longest common subsequence for adjacent intervals $S_1,S_2$ with total length $l < c\log n$ is less than $\frac{\varepsilon_1}{2}l$. We only need to consider pair of intervals $S_1,S_2$ whose total length $l\in [c\log n, 2c\log n]$.

Notice that the total length of $S_1$ and $S_2$ is at most $2c\log n$, which means that $S_1$ and $S_2$ each intersects with at most $3$ codewords from $\mathcal{C}$.
%According to the following result:
%\begin{lemma}[Lemma 4.3 of \cite{DBLP:journals/corr/abs-1710-07356}]
%Suppose $T_1$ is the concatenation of $\ell_1$ strings, $T_1=S_1\circ\dots\circ S_{\ell_1}$ and $T_2$ is the concatenation of $\ell_2$ strings, $T_2=S'_1\circ\dots\circ S'_{\ell_2}$. If  there exists an integer $t$ such that for all $i, j$, we have $LCS(S_i, S'_j) \leq t$, then we have $LCS(T_1, T_2) \leq (\ell_1+\ell_2)t$.
%\end{lemma}
Using Lemma \ref{algo:lcs},
we have that
$LCS(S_1,S_2)\leq 6\alpha l$.

Thus picking $\varepsilon = \max\{12\alpha, \varepsilon_1\} = 12\alpha = 12(1-\frac{1-\varepsilon_0}{1+\varepsilon_0}\delta)$, $S$ from algorithm \ref{algo:clonglineartime} is a $c$-long distance $\varepsilon$-synchronization circle.
\end{proof}
}
\fullOnly{\DetailedProofOfLemConstruction}

We need the following code constructed by Guruswarmi and Indyk \cite{1512415}. %with a small inner code obtained by brute-force search.

\begin{lemma}[Theorem 3 of \cite{1512415}]
\label{lem:outercode}
For every $  0< r < 1$, and all sufficiently small $\varepsilon>0$, there exists a family of codes of rate $r$ and relative distance $(1-r-\varepsilon)$ over an alphabet of size $2^{O(\varepsilon^{-4}r^{-1}\log(1/\varepsilon))}$ such that codes from the family can be encoded in linear time and can also be uniquely decoded in linear time from $2(1-r-\varepsilon)$ fraction of half errors.
\end{lemma}

\begin{lemma}[ECC by Brute Force Search]
\label{lem:innercode}

For any $n\in \mathbb{N}$, any $\varepsilon \in [0,1]$,
one can construct a ECC in time $  O(2^{\varepsilon n} (\frac{2e}{\varepsilon})^n n \log(1/\varepsilon))  $ and space $  O(2^{\varepsilon n} n \log (1/\varepsilon) ) $, with block length $n$, number of codewords $  2^{\varepsilon n}$, distance $ d = (1-\varepsilon)n$, alphabet size $2e/\varepsilon$.
\end{lemma}

\global\def\DetailedProofOfLemInnerCode{ 	% Define detailed proof of theorem
\shortOnly{\begin{proof}[Proof of Lemma~\ref{lem:innercode}]}
\fullOnly{\begin{proof}}

We conduct a brute-force search here to find all the codewords one by one.

We denote the code as $\mathcal{C}$ and the alphabet as $\Sigma$. Let $|\Sigma| = q$.
At first, let $\mathcal{C} = \emptyset$. Then we add an arbitrary element in $\Sigma^n$ to $\mathcal{C}$. Every time after a new element $C$ is added to $\mathcal{C}$, we exclude every such element in $\Sigma^{n}$ that has distance less than $d$ from $C$. Then we pick an arbitrary one from the remaining elements, adding it to $\mathcal{C}$. Keep doing this until $|\mathcal{C}| = 2^{\varepsilon n}$.

Note that given $C \in \Sigma^n$, the total number of elements that have distance less than $d$ to $C$, is at most ${n \choose d} q^d  = {n \choose (n-d)} q^d \leq (\frac{e}{\varepsilon})^{\varepsilon n} q^{(1-\varepsilon) n} $. We have to require that  $ |\mathcal{C}| (\frac{e}{\varepsilon})^{\varepsilon n} q^{(1-\varepsilon) n} \leq q^n $.
Let $q = 2e/\varepsilon$. So $\mathcal{C}$ can be $2^{\varepsilon n}$.

The exclusion operation takes time $  O((\frac{2e}{\varepsilon})^n n \log(1/\varepsilon)) $ as we have to
exhaustively search the space and for each word we have to compute it's hamming distance to the new added code word. Since there are  $  2^{\varepsilon n}$ code words, the time complexity is as stated.

We have to record those code words, so the space complexity is also as stated.
\end{proof}
}
\fullOnly{}

We can now use Algorithm \ref{algo:clonglineartime} to give a linear time construction of $c-long-distance$ $\varepsilon$-synchronization strings.

\begin{theorem}
\label{thm:linearConstruction}
For every  $n\in \N$ and any constant $0<\varepsilon < 1$, there is a deterministic construction of a $c=O(\varepsilon^{-2})-long-distance$ $\varepsilon$-synchronization string $S\in\Sigma^n$ where $|\Sigma|=O(\varepsilon^{-3})$, in time $O(n)$. Moreover, $S[i,i+\log n]$ can be computed in $O(\frac{\log n}{\varepsilon^2})$ time.
\end{theorem}
%
%\begin{proof}[Proof Sketch]
%
%We apply Algorithm \ref{algo:clonglineartime}. We need to show how to construct the error correcting code $\mathcal{\hat{C}}$ and the synchronization circle $SC$.
%
%Here $\mathcal{\hat{C}}$ with linear time encoding is constructed using the code in Lemma \ref{lem:outercode} as the outer code and the one in Lemma \ref{lem:innercode} as the inner code. The distance rate $\delta$ can be $1-\varepsilon' = \frac{1+\frac{\varepsilon}{36}}{1-\frac{\varepsilon}{36}}(1-\frac{\varepsilon}{12})$. The synchronization circle $SC$  is constructed by Corollary \ref{corollary:polyepssynccircle} in time $\poly(m) = \poly(\log n)$.
%\end{proof}

\global\def\DetailedProofOfThmLinearConstruction{ 	% Define detailed proof of theorem
\shortOnly{\begin{proof}[Proof of Theorem~\ref{thm:linearConstruction}]}
\fullOnly{\begin{proof}}
Suppose we have an error correcting code $\hat{\mathcal{C}}$ with distance  rate $1-\varepsilon' = \frac{1+\frac{\varepsilon}{36}}{1-\frac{\varepsilon}{36}}(1-\frac{\varepsilon}{12})$, message rate $r_c = O(\varepsilon'^{2})$, over an alphabet of size $|\Sigma_c| = O(\varepsilon'^{-1})$, with block length $m = O(\varepsilon'^{-2}\log n)$. Let $c = O(\varepsilon'^{-2}) =  O(\varepsilon^{-2})$.
%Concatenate $\hat{\mathcal{C}}$ with an $ \frac{\varepsilon}{36}$-synchronization circle $SC$ of length $m$ over an alphabet of size $O(\varepsilon^{-2})$.
We apply Algorithm \ref{algo:clonglineartime}, using $\hat{C}$ and an $ \frac{\varepsilon}{36}$-synchronization circle $SC$ of length $m$ over an alphabet of size $O(\varepsilon^{-2})$. Here $SC$ is constructed by Corollary \ref{corollary:polyepssynccircle} in time $\poly(m) = \poly(\log n)$.
By Lemma \ref{lem:construction}, we have a $c$-long-distance $12(1-\frac{1-\frac{\varepsilon}{36}}{1+\frac{\varepsilon}{36}}(1-\varepsilon'))=\varepsilon$-synchronization string of length $m\cdot |\Sigma_c|^{r_c m}\geq n$.

It remains to show that we can have such a  $\hat{\mathcal{C}}$ with linear time encoding.
We use the code in Lemma \ref{lem:outercode} as the outer code and the one in Lemma \ref{lem:innercode} as inner code. Let $\mathcal{C}_{out}$ be an instantiation of the code in Lemma \ref{lem:outercode} with rate $r_o = \varepsilon_o = \frac{1}{3} \varepsilon'$, relative distance $d_o = (1-2\varepsilon_o)$ and alphabet size $2^{O(\varepsilon_o^{-5}\log(1/\varepsilon_o))}$, and block length $n_o = \frac{\varepsilon_o^4\log n}{\log(1/\varepsilon_o)}$, which is encodable and decodable in linear time.

Further, according to Lemma \ref{lem:innercode} one can find a code $\mathcal{C}_{in}$ with rate $r_i = O(\varepsilon_i)$ where $\varepsilon_i = \frac{1}{3}\varepsilon'$, relative distance $1-\varepsilon_i$, over an alphabet of size $\frac{2e}{\varepsilon_i}$, and block length $n_i = O(\varepsilon_i^{-6}\log(1/\varepsilon_i))$. Note that since the block length and alphabet size are both constant because $\varepsilon$ is a constant. So the encoding can be done in constant time.

Concatenating $\mathcal{C}_{out}$ and $\mathcal{C}_{in}$ gives the desire code $\hat{\mathcal{C}}$ with rate $O(\varepsilon'^2)$, distance $1-O(\varepsilon')$ and alphabet of size $O(\varepsilon'^{-1})$ and block length $O(\varepsilon'^{-2}\log n)$. Moreover, the encoding of $\hat{\mathcal{C}}$ can be done in linear time, because the encoding of  $\mathcal{C}_{out}$ is in linear time and the encoding of $\mathcal{C}_{in}$ is in constant time.

Note that since every codeword of $\hat{\mathcal{C}}$ can be computed in time $O(\frac{\log n}{\varepsilon^2})$, $S[i,i+\log n]$ can be computed in $O(\frac{\log n}{\varepsilon^2})$ time.
\end{proof}
}
\fullOnly{\DetailedProofOfThmLinearConstruction}

%Since a $c$-long $\varepsilon$-synchronization string is also an $\varepsilon$-synchronization string, we directly have the following corollary.
\begin{corollary}
\label{corollary:linearepssyncstring}
For every $n\in \N$ and any constant $0<\varepsilon < 1$, there is a deterministic construction of an $\varepsilon$-synchronization string $S\in\Sigma^n$ where $|\Sigma|=O(\varepsilon^{-3})$, in time $O(n)$. Moreover, $S[i,i+\log n]$ can be computed in $O(\frac{\log n}{\varepsilon^2})$ time.
\end{corollary}

\subsection{Explicit Constructions of Infinite Synchronization Strings}\label{sec:infiniteConstruction}

In this section we focus on the construction of infinite synchronization strings.\

The definition of infinite synchronization strings is similar to that of standard synchronization strings except that the length of the string is infinite.
\begin{definition}[Infinite $\varepsilon$-synchronization string] A string $S$ is
an infinite $\varepsilon$-synchronization string if it has infinite length and $\forall 1\leq i< j < k$, $ED(S[i,j),S[j,k))>(1-\varepsilon)(k-i)$.
\end{definition}

To measure the efficiency of the construction of an infinite string, we consider the time complexity for computing the first $n$ elements of that string.
%We will show how to deterministically construct an infinite-long $\varepsilon$ synchronization string over an alphabet $\Sigma$ of size $O(\varepsilon^{-3})$. Our construction can compute the first $n$ elements of the infinite string in $O(n)$ time.
An infinite synchronization string is said to have an explicit construction if there is an algorithm that computes any position $S[i]$ in time $ \poly(i) $. Moreover,
it is said to have a highly-explicit construction if there is an algorithm that computes any position $S[i]$ in time $ O(\log i) $.

%\begin{theorem}
%For all $0<\varepsilon < 1$, there exists a construction of an infinite $\varepsilon$-synchronization string $S$  over a $O(\varepsilon^{-2})$-sized alphabet that is highly explicit. Moreover, $S[i,i+\log i]$ can be computed in $O(\log i)$ time.
%\end{theorem}

We have the following algorithm.

\begin{algorithm}[H]
\caption{Construction of infinite $\varepsilon$-synchronization string}
\label{algo:infinite}
\begin{algorithmic}
\STATE \textbf{Input:}
\begin{itemize}
\item A constant $\varepsilon \in (0,1)$.
\end{itemize}
\STATE \textbf{Operations:}
\begin{itemize}

\item Let $q \in \N$ be the size of an alphabet large enough to construct an $\frac{\varepsilon}{2}$-synchronization string. Let $\Sigma_1$ and $\Sigma_2$ be two alphabets of size $q$ such that $\Sigma_1 \cap \Sigma_2 =\emptyset$.
\item
%Construct a series of $\frac{\varepsilon}{2}$-synchronization strings in the following way:
Let $k=\frac{4}{\varepsilon}$. For $i = 1, 2, \ldots$, construct an $\frac{\varepsilon}{2}$-synchronization string  $S_{k^i}$  of length $k^i$, where $S_{k^i}$ is over $\Sigma_1$ if $i$ is odd and over $\Sigma_2$ otherwise.
\item Let $S$ be the sequential concatenation of $S_{k},S_{k^2},S_{k^3},\dots,S_{k^t},\dots$
\end{itemize}
\STATE \textbf{Output:} $S$.
\end{algorithmic}
\end{algorithm}

%Take two alphabets $\Sigma_1$ and $\Sigma_2$ with no intersection and each has size $q$. Let $k = \frac{4}{\varepsilon}$ and let $S_l$ denotes a $\frac{\varepsilon}{2}$-synchronization string of length $l$. The infinite $\varepsilon$-synchronization string $S$ is constructed by concatenating $S_{k},S_{k^2},S_{k^3},\dots,S_{k^t},\dots$, where $S_{k^i}$ is over $\Sigma_1$ if $i$ is odd and over $\Sigma_2$ otherwise. Then $S$ is over an alphabet of size $2q$.

\begin{lemma}
\label{lem:infiniteConstruction}
%If for any $0<\varepsilon<1$ and $n\in \N$, there exist $\varepsilon$-synchronization strings over alphabet $\Sigma$ of size $q$, then one can construct an infinite $\varepsilon$-synchronization string over an alphabet of size $2q$.

If there is a construction of $\frac{\varepsilon}{2}$-synchronization strings with alphabet size $q$, then Algorithm \ref{algo:infinite} constructs an infinite $\varepsilon$-synchronization string owith alphabet size $2q$.
\end{lemma}

%\begin{proof}[Proof Sketch]
%We first show the following claim.
%\begin{claim*}
%\label{claim:infinite}
%Let $x<y<z$ be positive integers and let $t$ be such that $k^t\leq |S[x,z)|<k^{t+1}$. Then $ED(S[x,y),S[y,z))\geq (1-\frac{\varepsilon}{2})(z-x) (1-\frac{2}{k}) $.
%\end{claim*}
%The main observation used to prove the claim is that we only have to consider the last (longest) two blocks that the interval intersects. 
%
%Since $k = \frac{4}{\eps}$, by the claim,
%\[ED(S[x,y),S[y,z)) \geq \left(1-\frac{\varepsilon}{2}\right)(z-x)\left(1-\frac{2}{k}\right)=\left(1-\frac{\varepsilon}{2}\right)^2(z-x)\geq(1-\varepsilon)(z-x).\]
%
%As a result, $S$ is an $\varepsilon$-synchronization string.
%\end{proof}

\global\def\DetailedProofOfLemInfinite{ 	% Define detailed proof of theorem
\shortOnly{\begin{proof}[Proof of Lemma~\ref{lem:infiniteConstruction}]}
\fullOnly{\begin{proof}}

Algorithm \ref{algo:infinite} can be shown as in the figure below.

\begin{figure}[H]
  \centering
  \includegraphics[width=8cm]{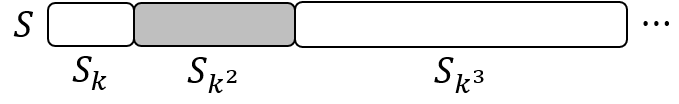}\\
  \caption{$S_k$ and $S_{k^3}$ are over alphabet $\Sigma_1$ and $S_{k^2}$ is over $\Sigma_2$.}
\end{figure}

%First, if each $S_{k^i}$ is constructed through the highly explicit way in Theorem \ref{thm:linearConstruction}, any substring of length $n$ can be computed in $O(n)$ time. Further, any substring $S[i,i+\log i]$ can be computed in $O(\log i)$. All that remains now is to show that $S$ is an $\varepsilon$-synchronization string. We use the following lemma.

Now we show that $S$ is an infinite $\varepsilon$-synchronization string.

\begin{claim*}
Let $x<y<z$ be positive integers and let $t$ be such that $k^t\leq |S[x,z)|<k^{t+1}$. Then $ED(S[x,y),S[y,z))\geq (1-\frac{\varepsilon}{2})(z-x) (1-\frac{2}{k}) $.
\end{claim*}

\begin{proof}

Let $l_i$ be the index of $S$ where $S_{k^{i+1}}$ starts. Then $l_i = \sum_{j=1}^i k^j = \frac{k^{i+1}-k}{k-1}$. Notice that $l_{t-1}<2k^{t-1}$ and $|S[x,z)|\geq k^t $, one can throw away all elements of $S[x,z)$ whose indices are less than $l_{t-1}$ without losing more than $ \frac{2k^{t-1}}{k^t} = \frac{2}{k}$ fraction of the elements of $S[x,z)$. We use $S[x',z)$ to denote the substring after throwing away the symbols before $l_{t-1}$. Thus $x'\geq l_{t-1}$.

Since $x'\geq l_{t-1}$, $S[x',z)$ either entirely falls into a synchronization string $S_{k^l}$ or crosses two synchronization strings $S_{k^l}$ and $S_{k^{l+1}}$ over two entirely different alphabets $\Sigma_1$ and $\Sigma_2$. Thus the edit distance of $S[x',y)$ and $S[y,z)$ is at least $(1-\frac{\varepsilon}{2})(z-x)$.
\end{proof}

Since $k = \frac{4}{\eps}$, we have that \[ED(S[x,y),S[y,z)) \geq \left(1-\frac{\varepsilon}{2}\right)(z-x)\left(1-\frac{2}{k}\right)=\left(1-\frac{\varepsilon}{2}\right)^2(z-x)\geq(1-\varepsilon)(z-x).\]

This shows that $S$ is an $\varepsilon$-synchronization string.
\end{proof}
}
\fullOnly{\DetailedProofOfLemInfinite}

If we instantiate Algorithm \ref{algo:infinite} using Corollary \ref{corollary:polyepssyncstring}, then we have the following theorem.
\begin{theorem}
\label{thm:polyinfinite}
For any constant $0<\varepsilon < 1$, there exists an explicit construction of an infinite $\varepsilon$-synchronization string $S$  over an alphabet of size $O(\varepsilon^{-2})$.
\end{theorem}

\global\def\DetailedProofOfThmPolyInfinite{ 	% Define detailed proof of theorem
\shortOnly{\begin{proof}[Proof of Theorem~\ref{thm:polyinfinite}]}
\fullOnly{\begin{proof}}

We combine Algorithm \ref{algo:infinite} and Corollary \ref{corollary:polyepssyncstring}.
In the algorithm, we can construct every substring $S_{k^i}$ in polynomial time with alphabet size $q=O(\varepsilon^{-2})$, by Corollary \ref{corollary:polyepssyncstring}. So the first $n$ symbols of $S$ can be computed in polynomial time.

By Lemma \ref{lem:infiniteConstruction}, $S$ is an infinite $\varepsilon$-synchronization string over an alphabet of size $2q=O(\varepsilon^{-2})$.
\end{proof}
}
\fullOnly{\DetailedProofOfThmPolyInfinite}

If we instantiate Algorithm \ref{algo:infinite} using Corollary \ref{corollary:linearepssyncstring}, then we have the following theorem.
\begin{theorem}
\label{thm:LinearInfinite}
For any constant $0<\varepsilon < 1$, there exists a highly-explicit construction of an infinite $\varepsilon$-synchronization string $S$ over an alphabet of size $O(\varepsilon^{-3})$.  Moreover, for any $i\in \mathbb{N}$, the first $i$ symbols can be computed in $O(i)$ time and $S[i,i+\log i]$ can be computed in $O(\log i)$ time.
\end{theorem}

\global\def\DetailedProofOfThmLinearInfinite{ 	% Define detailed proof of theorem
\shortOnly{\begin{proof}[Proof of Theorem~\ref{thm:LinearInfinite}]}
\fullOnly{\begin{proof}}
Combine Algorithm \ref{algo:infinite} and Corollary \ref{corollary:linearepssyncstring}.
In the algorithm, we can construct every substring $S_{k^i}$ in linear time with alphabet size $q=O(\varepsilon^{-3})$, by Corollary \ref{corollary:linearepssyncstring}. So  the first $i$ symbols can be computed in $O(i)$ time. Also any substring $S[i,i+\log i]$ can be computed in time $O(\log i)$.

By Lemma \ref{lem:infiniteConstruction}, $S$ is an infinite $\varepsilon$-synchronization string over an alphabet of size $2q=O(\varepsilon^{-3})$.
\end{proof}
}
\fullOnly{\DetailedProofOfThmLinearInfinite}

%% file: constalphabet.tex
\section{$\Omega\left(\eps^{-3/2}\right)$ Lower-Bound on Alphabet Size}
The \emph{twin word} problem was introduced by Axenovich, Person, and Puzynina~\cite{axenovich2013regularity} and further studied by Bukh and Zhou~\cite{bukh2016twins}.\ Any set of two identical disjoint subsequences in a given string is called a twin word. \cite{axenovich2013regularity, bukh2016twins} provided a variety of results on the relations between the length of a string, the size of the alphabet over which it is defined, and the size of the longest twin word it contains. We will make use of the following result from~\cite{bukh2016twins} that is built upon Lemma 5.9 from~\cite{beame2008value} to provide a new lower-bound on the alphabet size of synchronization strings.
\begin{theorem}[Theorem 3 from~\cite{bukh2016twins}]
There exists a constant $c$ so that every word of length $n$ over a $q$-letter alphabet contains two disjoint equal subsequences of length $cnq^{-2/3}$.
\end{theorem}

Further, Theorem 6.4 of~\cite{haeupler2017synchronization} states that any $\eps$-synchronization string of length $n$ has to satisfy \emph{$\eps$-self-matching property} which essentially means that it cannot contain two (not necessarily disjoint) subsequences of length $\eps n$ or more. These two requirements lead to the following inequality for an $\eps$-synchronization string of length $n$ over an alphabet of size $q$.
$$cnq^{-2/3} \le \eps n \Rightarrow c'\eps^{-3/2} \le q$$

\section{Synchronization Strings over Small Alphabets}
\label{sec:syncexsmallalphabets}
In this section, we focus on synchronization strings over small constant-sized alphabets.
We study the question of what is the smallest possible alphabet size over which arbitrarily long $\varepsilon$-synchronization strings can exist for some $\varepsilon < 1$, and how such synchronization strings can be constructed.

Throughout this section, we will make use of square-free strings introduced by Thue~\cite{thue1977unendliche}, which is a weaker notion than synchronization strings that requires all consecutive equal-length substrings to be non-identical. Note that no synchronization strings or square-free strings of length four or more exist over a binary alphabet since a binary string of length four either contains two consecutive similar symbols or two identical consecutive substrings of length two. However, for ternary alphabets, arbitrarily long square-free strings exist and can be constructed efficiently using \emph{uniform morphism}~\cite{zolotov2015another}. In Section~\ref{sec:morphism}, we will briefly review this construction and show that no uniform morphism can be used to construct arbitrary long synchronization strings. In Section~\ref{sec:alphabetFour}, we make use of ternary square-free strings to show that arbitrarily long $\varepsilon$-synchronization strings exist over alphabets of size four for some $\varepsilon<1$. Finally, in Section~\ref{sec:epsForAlphabet}, we provide experimental lower-bounds on $\varepsilon$' for which $\varepsilon$-synchronization strings exist over alphabets of size 3, 4, 5, and 6.

\subsection{Morphisms cannot Generate Synchronization Strings}\label{sec:morphism}
Previous works show that one can construct infinitely long square-free or approximate-square-free strings using \emph{uniform morphisms}. A uniform morphism of rank $r$ over an alphabet $\Sigma$ is a function $\phi:\Sigma\rightarrow\Sigma^r$ that maps any symbol out of an alphabet $\Sigma$ to a string of length $r$ over the same alphabet. Applying the function $\phi$ over some string $S\in\Sigma^*$ is defined as replacing each symbol of $S$ with $\phi(S)$.

\cite{krieger2007avoiding, thue1912gegenseitige, leech19572726, crochemore1982sharp, zolotov2015another} show that there are uniform morphisms that generate the combinatorial objects they study respectively. More specifically, one can start from any letter of the alphabet and repeatedly apply the morphism on it to construct those objects. For instance, using the uniform morphisms of rank 11 suggested in \cite{zolotov2015another}, all such strings will be square-free. In this section, we investigate the possibility of finding similar constructions for synchronization strings. We will show that no such morphism can possibly generate an infinite $\varepsilon$-synchronization strings for any fixed $0<\varepsilon<1$.

The key to this claim is that a matching between two substrings is preserved under an application of the uniform morphism $\phi$. Hence, we can always increase the size of a matching between two substrings by applying the morphism sufficiently many times, and then adding new matches to the matching from previous steps.

\begin{theorem}\label{thm:MorphismsImpossibility}
Let $\phi$ be a uniform morphism of rank $r$ over alphabet $\Sigma$. Then $\phi$ does not generate an infinite $\varepsilon$-synchronization string, for any $0< \varepsilon <  1$.
\end{theorem}
\global\def\DetailedProofOfThmMorphismImpossibility{ 	% Define detailed proof of theorem
\shortOnly{\begin{proof}[Proof of Theorem~\ref{thm:MorphismsImpossibility}]}
\fullOnly{\begin{proof}}
To prove this, we show that for any $0<\varepsilon<1$, applying morphism $\phi$ sufficiently many times over any symbol of alphabet $\Sigma$ produces a strings that has two neighboring intervals which contradict $\varepsilon$-synchronization property.
First, we claim that, without loss of generality, it suffices to prove this for morphisms $\phi$ for which $\phi(\sigma)$ contains all elements of $\Sigma$ for any $\sigma \in \Sigma$.
To see this, consider the graph $G$ with $|\Sigma|$ vertices where each vertex corresponds to a letter of the alphabet and there is a $(\sigma_1, \sigma_2)$ edge if $\phi(\sigma_1)$ contains $\sigma_2$. It is straightforward to verify that after applying morphism $\phi$ over a letter sufficiently many times, the resulting string can be split into a number of substrings so that the symbols in any of them belong to a subset of $\Sigma$ that corresponds to some strongly connected component in $G$. As $\varepsilon$-synchronization string property is a hereditary property over substrings, this gives that one can, without loss of generality, prove the above-mentioned claim for morphisms $\phi$ for which the corresponding graph $G$ is strongly connected. Further, let $d$ be the greatest common divisor of the size of all cycles in $G$. One can verify that, for some sufficiently large $k$, $\phi^{kd}$ will be a morphism that, depending on the letter $\sigma$ to perform recursive applications of the morphism on, will always generate strings over some alphabet $\Sigma_{\sigma}$ and $\phi^{kd}(\sigma')$ contains all symbols of $\Sigma_{\sigma}$ for all $\sigma'\in\Sigma_{\sigma}$. As proving the claim for $\phi^{kd}$ implies it for $\phi$ as well, the assumption mentioned above does not harm the generality.

We now proceed to prove that for any morphism $\phi$ of rank $r$ as described above, any positive integer $n\in\N$, and any positive constant $0<\delta<1$, there exists $m\in\N$ so that
$$LCS(\phi^m(a), \phi^m(b)) \ge \left[1 - \left(1 - \frac{1}{|\Sigma|^2r}\right)^n - \delta\right]\cdot r^m$$
 for any $a, b\in \Sigma$ where $\phi^m$ represents $m$ consecutive applications of morphism $\phi$ and $LCS(., .)$ denotes the longest common substring.

Having such claim proved, one can take $\delta=(1-\varepsilon)/2$ and $n$ large enough so that $m$ applications of $\phi$ over any pair of symbols entail strings with a longest common substring that is of a fraction larger than $1-(1-\varepsilon)=\varepsilon$ in terms of the length of those strings. Then, for any string $S\in\Sigma^*$, one can take two arbitrary consecutive symbols of $\phi(S)$ like $S[i]$ and $S[i+1]$. Applying morphism $\phi$, $m$ more times on $\phi(S)$ makes the corresponding intervals of $\phi^{m+1}(S)$ have an edit distance that is smaller than $1-\varepsilon$ fraction of their combined lengths. This shows that $\phi^{m+1}(S)$ is not an $\varepsilon$-synchronization string and finishes the proof.

Finally, we prove the claim by induction on $n$. For the base case of $n=1$, given the assumption of all members of $\Sigma$ appearing in $\phi(\sigma)$ for all $\sigma\in\Sigma$, $\phi(a)$ and $\phi(b)$ have a non-empty common subsequence. This gives that
$$LCS(\phi(a), \phi(b)) \ge 1 = \left[1-\left(1-\frac{1}{r}\right)\right] \cdot r > \left[1-\left(1-\frac{1}{|\Sigma|^2r}\right)-\delta\right] \cdot r.$$
Therefore, choosing $m=1$ finishes the induction base.

We now prove the induction step. Note that by induction hypothesis, for some given $n$, one can find $m_1$ such that
$$LCS(\phi^{m_1}(a), \phi^{m_1}(b)) \ge \left[1 - \left(1 - \frac{1}{|\Sigma|^2r}\right)^n - \frac{\delta}{2}\right]\cdot r^{m_1}.$$
Now, let $m_2=\left\lceil \log_r \frac{2}{\delta} \right\rceil$. Consider $\phi^{m_2}(a)$ and $\phi^{m_2}(b)$. Note that among all possible pairs of symbols from $\Sigma^2$, one appears at least $\frac{r^{m_2}}{|\Sigma|^2}$ times in respective positions of $\phi^{m_2}(a)$ and $\phi^{m_2}(b)$. Let $(a', b')$ be such pair. As $\phi(a')$ and $\phi(b')$ contain all symbols of $\Sigma$, one can take one specific occurrence of a fixed arbitrary symbol $\sigma\in\Sigma$ in all appearances of the pair $\phi(a')$ and $\phi(b')$ to find a common subsequence of size
$\frac{r^{m_2}}{|\Sigma|^2}=\frac{r^{m_2+1}}{|\Sigma|^2r}$ or more
between $\phi^{m_2+1}(a)$ and $\phi^{m_2+1}(b)$ (See Figure~\ref{fig:morphism}).
\begin{figure}[H]
  \centering
  \includegraphics[width=12.5cm]{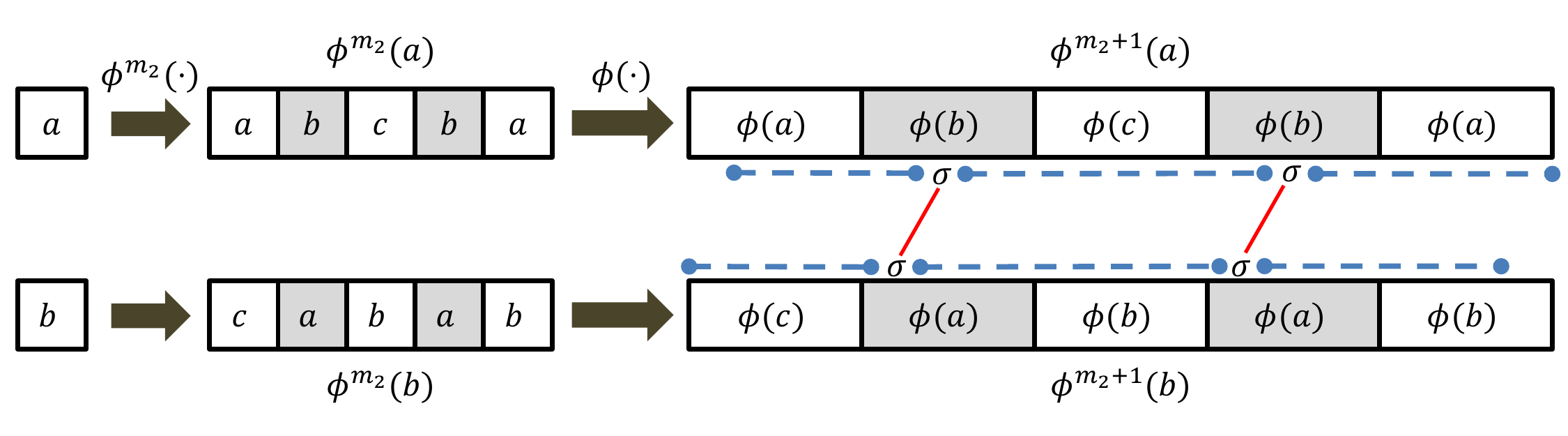}\\
  \caption{Induction step in Theorem~\ref{thm:MorphismsImpossibility}; Most common pair $(a', b')=(b, a)$.}\label{fig:morphism}
\end{figure}
Note that one can apply the morphism $\phi$ further times over $\phi^{m_2+1}(a)$ and $\phi^{m_2+1}(b)$ and such common subsequence will still be preserved; However, one might be able to increase the size of it by adding new elements to the common subsequence from equal length pairs of intervals between current common subsequence elements (denoted by blue dashed line in Figure~\ref{fig:morphism}). The total length  of such intervals is
$$1-\frac{1}{|\Sigma|^2r}-\frac{r}{r^{m_2+1}} = 1-\frac{1}{|\Sigma|^2r}-\frac{\delta}{2}$$
or more.
In fact, using the induction hypothesis, by applying the morphism $m_1$ more times, one can get the following for $m=m_1+m_2+1$.
\begin{eqnarray*}
LCS(\phi^{m}(a), \phi^{m}(b)) &\ge&
\Bigg[
\frac{1}{|\Sigma|^2r}+
\left(1 - \left(1 - \frac{1}{|\Sigma|^2r}\right)^n - \frac{\delta}{2}\right)\\
&& \cdot \left(1-\frac{1}{|\Sigma|^2r}-\frac{\delta}{2}\right)
\Bigg] r^{m}\\
&\ge&
\left[1 - \left(1 - \frac{1}{|\Sigma|^2r}\right)^{n+1} - \delta\right] r^{m}
\end{eqnarray*}
This completes the induction step and finishes the proof.
\end{proof}
}
\fullOnly{\DetailedProofOfThmMorphismImpossibility}

% ============================
\subsection{Synchronization Strings over Alphabets of Size Four}\label{sec:alphabetFour}
In this section, we show that synchronization strings of arbitrary length exist over alphabets of size four. In order to do so, we first introduce the notion of \emph{weak $\varepsilon$-synchronization strings}. This weaker notion is very similar to the synchronization string property except the edit distance requirement is rounded down.

\begin{definition}[weak $\varepsilon$-synchronization strings]
String $S$ of length $n$ is a weak $\varepsilon$-synchronization string if for every $1\le i<j<k\le n$, $$ED(S[i, j), S[j, k)) \ge \lfloor(1-\varepsilon)(k-i)\rfloor.$$
\end{definition}
We start by showing that binary weak $\varepsilon$-synchronization strings exist for some $\varepsilon < 1$.

\subsubsection{Binary Weak $\varepsilon$-Synchronization Strings}
Here we prove that an infinite binary weak $\varepsilon$-synchronization string exists.
The main idea is to take a synchronization string over some large alphabet and convert it to a binary weak synchronization string by mapping each symbol of that large alphabet to a binary string and separating each binary encoded block with a block of the form $0^k1^k$.

\begin{theorem}\label{weaksync}
There exists a constant $\varepsilon < 1$ and an infinite binary weak $\varepsilon$-synchronization string.
\end{theorem}

\global\def\DetailedProofOfThmWeaksync{ 	% Define detailed proof of theorem
\shortOnly{\begin{proof}[Proof of Theorem~\ref{weaksync}]}
\fullOnly{\begin{proof}}
Take some arbitrary $\varepsilon'\in(0, 1)$. According to \cite{haeupler2017synchronization}, there exists an infinite $\varepsilon'$-synchronization string $S$ over a sufficiently large alphabet $\Sigma$. Let $k = \lceil \log |\Sigma| \rceil$. Translate each symbol of $S$ into $k$ binary bits, and separate the translated $k$-blocks with $0^k 1^k$. We claim that this new string $T$ is a weak $\varepsilon$-synchronization binary string for some $\varepsilon < 1$.

First, call a translated $k$-length symbol followed by $0^k1^k$ a \emph{full block}. Call any other (possibly empty) substring a \emph{half block}. Then any substring of $T$ is a half-block followed by multiple full blocks and ends with a half block.

Let $A$ and $B$ be two consecutive substrings in $T$. Without loss of generality, assume $|A| \leq |B|$ (because edit distance is symmetric). Let $M$ be a longest common subsequence between $A$ and $B$. Partition blocks of $B$ into the following 4 types of blocks:
\begin{enumerate}
\item Full blocks that match completely to another full block in $A$.
\item Full blocks that match completely but not to just 1 full block in $A$.
\item Full blocks where not all bits within are matched.
\item Half blocks.
\end{enumerate}

The key claim is that the $3k$ elements in $B$ which are matched to a type-2 block in $A$ are not contiguous and, therefore, there is at least one unmatched symbol in $B$ surrounded by them. To see this, assume by contradiction that all letters of some type-2 block in $A$ are matched contiguously. The following simple analysis over 3 cases contradicts this assumption:
\begin{itemize}
\item Match starts at middle of some translated $k$-length symbol, say position $p \in [2,k]$. Then the first $1$ of $1^k$ in $A$ will be matched to the $(k-p+2)$-th $0$ of $0^k$ in $B$, contradiction.
\item Match starts at $0$-portion of $0^k1^k$ block, say at the $p$-th $0$. Then the $p$-th $1$ of $1^k$ in $A$ will be matched to the first $0$ of $0^k$ in $B$, contradiction.
\item Match starts at $1$-portion of $0^k1^k$ block, say at the $p$-th $1$. Then the $p$-th $0$ of $0^k$ in $A$ will be matched to the first $1$ of $1^k$ in $B$, contradiction.
\end{itemize}

Let the number of type-i blocks in $B$ be $t_i$. For every type-2 block, there is an unmatched letter in $A$ between its first and last matches. Hence, $|A| \geq |M| + t_2$. For every type-3 block, there is an unmatched letter in $B$ within. Hence, $|B| \geq |M| + t_3$. Therefore, $|A| + |B| \geq 2|M| + t_2 + t_3$.

Since $|A| \leq |B|$, the total number of full blocks in both $A$ and $B$ is at most $2(t_1+t_2+t_3)+1$. (the additional $+1$ comes from the possibility that the two half-blocks in $B$ allows for one extra full block in $A$) Note $t_1$ is a matching between the full blocks in $A$ and the full blocks in $B$. So due to the $\varepsilon'$-synchronization property of $S$, we obtain the following.
\[ t_1 \leq \frac{\varepsilon'}{2}\left(2(t_1+t_2+t_3) + 1\right) \implies t_1 \leq \frac{\varepsilon'}{1-\varepsilon'}(t_2+t_3) + \frac{\varepsilon'}{2(1-\varepsilon')}\]
Furthermore, $t_1+t_2+t_3+2 > \frac{|B|}{3k} \geq \frac{|A|+|B|}{6k}$. This, along with the above inequality, implies the following.
\[ \frac{1}{1-\varepsilon'}(t_2 + t_3) + \frac{4-3\varepsilon'}{2(1-\varepsilon')} > \frac{|A|+|B|}{6k}. \]
The edit distance between $A$ and $B$ is
\begin{eqnarray*}
ED(A,B) &=& |A| + |B| - 2|M| \geq t_2 + t_3 \\
&>& \frac{1-\varepsilon'}{6k}(|A|+|B|) - \frac{4-3\varepsilon'}{2} > \frac{1-\varepsilon'}{6k}(|A|+|B|) - 2.
\end{eqnarray*}
Set $\varepsilon = 1-\frac{1-\varepsilon'}{18k}$. If $|A| + |B| \geq \frac{1}{1-\varepsilon}$, then
\begin{eqnarray*}
\frac{1-\varepsilon'}{6k}(|A|+|B|) - 2 &\geq& \left(\frac{1-\varepsilon'}{6k} - 2(1-\varepsilon)\right)(|A|+|B|)\\
 &=& (1-\varepsilon)(|A| + |B|) \geq \lfloor (1-\varepsilon)(|A| + |B|) \rfloor.
\end{eqnarray*}
As weak $\varepsilon$-synchronization property trivially holds for $|A| + |B| < \frac{1}{1-\varepsilon}$, this will prove that $T$ is a weak $\varepsilon$-synchronization string.
\end{proof}
}
\fullOnly{\DetailedProofOfThmWeaksync}

\subsubsection{$\varepsilon$-Synchronization Strings over Alphabets of Size Four}

A corollary of Theorem~\ref{weaksync} is the existence of infinite synchronization strings over alphabets of size four. Here we make use of the fact that infinite ternary square-free strings exist, which was proven in previous work~\cite{thue1977unendliche}. We then modify such a string to fulfill the synchronization string property, using the existence of an infinite binary weak synchronization string.

\begin{theorem}\label{thm:AlphabetFourSyncStr}
There exists some $\varepsilon \in (0,1)$ and an infinite $\varepsilon$-synchronization string over an alphabet of size four.
\end{theorem}

\global\def\DetailedProofOfThmAlphabetFourSyncStr{ 	% Define detailed proof of theorem
\shortOnly{\begin{proof}[Proof of Theorem~\ref{thm:AlphabetFourSyncStr}]}
\fullOnly{\begin{proof}}
Take an infinite ternary square-free string $T$ over alphabet $\{1, 2, 3\}$~\cite{thue1977unendliche} and some $\varepsilon \in \left(\frac{11}{12}, 1\right)$. Let $S$ be an infinite weak binary $\varepsilon'=(12\varepsilon-11)$-synchronization string. Consider the string $W$ that is similar to $T$ except that the $i$-th occurrence of symbol $1$ in $T$ is replaced with symbol $4$ if $S[i]=1$. Note $W$ is still square-free. We claim $W$ is an $\varepsilon$-synchronization string as well.

Let $A = W[i,j), B = W[j,k)$ be two consecutive substrings of $W$. If $k - i < 1/(1-\varepsilon)$, then $ED(A,B) \geq 1 > (1-\varepsilon)(k-i)$ by square-freeness.

Otherwise, $k - i \geq 1/(1-\varepsilon) \geq 12$. Consider all occurrences of $1$ and $4$ in $A$ and $B$, which form consecutive subsequences $A_s$ and $B_s$ of $S$ respectively. Note that $|A_s| + |B_s| \geq (k-i-3)/4$, because, by square-freeness, there cannot be a length-4 substring consisting only of $2$'s and $3$'s in $W$.

By weak synchronization property,
\begin{eqnarray*}
ED(A_s,B_s) &\geq& \lfloor (1-\varepsilon')(|A_s|+|B_s|)\rfloor\\
&\geq& \lfloor 3(1-\varepsilon)(k-i-3)\rfloor > 3(1-\varepsilon)(k-i)-9(1-\varepsilon) - 1\\ &\geq& (1-\varepsilon)(k-i),
\end{eqnarray*}
and hence, $ED(A,B) \geq ED(A_s,B_s) \geq (1-\varepsilon)(k-i)$. Therefore, $W$ is an $\varepsilon$-synchronization string.
\end{proof}
}
\fullOnly{\DetailedProofOfThmAlphabetFourSyncStr}

%% file: appendix.tex
\appendix
\begin{center}
\bfseries \huge Appendices
\end{center}

\shortOnly{
\section{Proofs in Section \ref{sec:synccircle}}

\DetailedProofOfThmSyncStr

\DetailedProofOfLemRandomAlgo

\DetailedProofOfThmSyncCircle

\section{Proofs in Section \ref{sec:LongDistSync}}

\DetailedProofOfThmPolyConstruction

\DetailedProofOfalgolcsLemma

\DetailedProofOfMainLemma

\DetailedProofOfLemConstruction

\DetailedProofOfLemInnerCode

\DetailedProofOfThmLinearConstruction

\DetailedProofOfLemInfinite

\DetailedProofOfThmPolyInfinite

\DetailedProofOfThmLinearInfinite

\section{Proofs in Section \ref{sec:syncexsmallalphabets}}
\DetailedProofOfThmMorphismImpossibility

\DetailedProofOfThmWeaksync

\DetailedProofOfThmAlphabetFourSyncStr

}
\section{Lower-bounds for $\varepsilon$ in Infinite $\varepsilon$-Synchronization Strings}\label{sec:epsForAlphabet}

It is known from Section~\ref{sec:alphabetFour} that infinite synchronization strings exist over alphabet sizes $|\Sigma| \geq 4$. A natural question to ask is the optimal value of $\varepsilon$ for each such $|\Sigma|$. Formally, we seek to discover
\[ B_k = \inf \{\varepsilon : \text{there exists an infinite $\varepsilon$-synchronization string with $|\Sigma| = k$}\} \]
for small values of $k$. To that end, a program was written to find an upper bound for $B_k$ for $k \leq 6$. The program first fixes an $\varepsilon$, then exhaustively enumerates all possible $\varepsilon$-synchronization strings over an alphabet size of $k$ by increasing length. If the program terminates, then this $\varepsilon$ is a proven lower bound for $B_k$. Among every pair of consecutive substrings in each checked string that failed the $\varepsilon$-synchronization property, we find the one that has the lowest edit distance relative to their total length and such fraction would be a lower-bound for $B_k$ as well. Such experimentally obtained lower-bounds for alphabets of size 3, 4, 5, and 6 are listed in Table~\ref{bkbounds}.
\begin{table}[h]
\begin{center}
\begin{tabular}{| c | c |}
\hline
$k$ & $B_k \geq \cdot$ \\ \hline
$3$ & $12/13$ \\
$4$ & $10/13$ \\
$5$ & $2/3$ \\
$6$ & $18/29$ \\ \hline
\end{tabular}
\end{center}
\vspace{-2ex}
\caption{Computationally proven lower-bounds of $B_k$}
\label{bkbounds}
\end{table}

%% file: main.bbl
\begin{thebibliography}{10}

\bibitem{axenovich2013regularity}
Maria Axenovich, Yury Person, and Svetlana Puzynina.
\newblock A regularity lemma and twins in words.
\newblock {\em Journal of Combinatorial Theory, Series A}, 120(4):733--743,
  2013.

\bibitem{beame2008value}
Paul Beame and Dang-Trinh Huynh-Ngoc.
\newblock On the value of multiple read/write streams for approximating
  frequency moments.
\newblock In {\em Foundations of Computer Science, 2008. FOCS'08. IEEE 49th
  Annual IEEE Symposium on}, pages 499--508. IEEE, 2008.

\bibitem{braverman2017coding}
Mark Braverman, Ran Gelles, Jieming Mao, and Rafail Ostrovsky.
\newblock Coding for interactive communication correcting insertions and
  deletions.
\newblock {\em IEEE Transactions on Information Theory}, 63(10):6256--6270,
  2017.

\bibitem{bukh2016twins}
Boris Bukh and Lidong Zhou.
\newblock Twins in words and long common subsequences in permutations.
\newblock {\em Israel Journal of Mathematics}, 213(1):183--209, 2016.

\bibitem{camungol2016avoiding}
Serina Camungol, Narad Rampersad, et~al.
\newblock Avoiding approximate repetitions with respect to the longest common
  subsequence distance.
\newblock {\em inv lve}, page 657, 2016.

\bibitem{chandrasekaran2013deterministic}
Karthekeyan Chandrasekaran, Navin Goyal, and Bernhard Haeupler.
\newblock Deterministic algorithms for the lov{\'a}sz local lemma.
\newblock {\em SIAM Journal on Computing}, 42(6):2132--2155, 2013.

\bibitem{crochemore1982sharp}
Max Crochemore.
\newblock Sharp characterizations of squarefree morphisms.
\newblock {\em Theoretical Computer Science}, 18(2):221--226, 1982.

\bibitem{gelles2015coding}
Ran Gelles.
\newblock Coding for interactive communication: A survey, 2015.

\bibitem{gelles2015capacity}
Ran Gelles and Bernhard Haeupler.
\newblock Capacity of interactive communication over erasure channels and
  channels with feedback.
\newblock In {\em Proceedings of the Twenty-Sixth Annual ACM-SIAM Symposium on
  Discrete Algorithms}, pages 1296--1311. Society for Industrial and Applied
  Mathematics, 2015.

\bibitem{ghaffari2014optimal2}
Mohsen Ghaffari and Bernhard Haeupler.
\newblock Optimal error rates for interactive coding ii: Efficiency and list
  decoding.
\newblock In {\em Foundations of Computer Science (FOCS), 2014 IEEE 55th Annual
  Symposium on}, pages 394--403. IEEE, 2014.

\bibitem{ghaffari2014optimal1}
Mohsen Ghaffari, Bernhard Haeupler, and Madhu Sudan.
\newblock Optimal error rates for interactive coding i: Adaptivity and other
  settings.
\newblock In {\em Proceedings of the forty-sixth annual ACM symposium on Theory
  of computing}, pages 794--803. ACM, 2014.

\bibitem{1512415}
V.~Guruswami and P.~Indyk.
\newblock Linear-time encodable/decodable codes with near-optimal rate.
\newblock {\em IEEE Transactions on Information Theory}, 51(10):3393--3400, Oct
  2005.
\newblock \href {http://dx.doi.org/10.1109/TIT.2005.855587}
  {\path{doi:10.1109/TIT.2005.855587}}.

\bibitem{haeupler2014interactive}
Bernhard Haeupler.
\newblock Interactive channel capacity revisited.
\newblock In {\em Foundations of Computer Science (FOCS), 2014 IEEE 55th Annual
  Symposium on}, pages 226--235. IEEE, 2014.

\bibitem{haeupler2011new}
Bernhard Haeupler, Barna Saha, and Aravind Srinivasan.
\newblock New constructive aspects of the lov{\'a}sz local lemma.
\newblock {\em Journal of the ACM (JACM)}, 58(6):28, 2011.

\bibitem{haeupler2017synchronization}
Bernhard Haeupler and Amirbehshad Shahrasbi.
\newblock Synchronization strings: codes for insertions and deletions
  approaching the singleton bound.
\newblock In {\em Proceedings of the 49th Annual ACM SIGACT Symposium on Theory
  of Computing}, pages 33--46. ACM, 2017.

\bibitem{HS17c}
Bernhard Haeupler and Amirbehshad Shahrasbi.
\newblock Synchronization strings: Explicit constructions, local decoding, and
  applications.
\newblock {\em arXiv preprint arXiv:1710.09795}, 2017.

\bibitem{haeupler2018synchronization}
Bernhard Haeupler, Amirbehshad Shahrasbi, and Madhu Sudan.
\newblock Synchronization strings: List decoding for insertions and deletions,
  2018.

\bibitem{haeupler2017synsimucode}
Bernhard Haeupler, Amirbehshad Shahrasbi, and Ellen Vitercik.
\newblock Synchronization strings: Channel simulations and interactive coding
  for insertions and deletions.
\newblock {\em arXiv preprint arXiv:1707.04233}, 2017.

\bibitem{kol2013interactive}
Gillat Kol and Ran Raz.
\newblock Interactive channel capacity.
\newblock In {\em Proceedings of the forty-fifth annual ACM symposium on Theory
  of computing}, pages 715--724. ACM, 2013.

\bibitem{krieger2007avoiding}
Dalia Krieger, Pascal Ochem, Narad Rampersad, and Jeffrey Shallit.
\newblock Avoiding approximate squares.
\newblock In {\em International Conference on Developments in Language Theory},
  pages 278--289. Springer, 2007.

\bibitem{leech19572726}
John Leech.
\newblock 2726. a problem on strings of beads.
\newblock {\em The Mathematical Gazette}, 41(338):277--278, 1957.

\bibitem{mercier2010survey}
Hugues Mercier, Vijay~K Bhargava, and Vahid Tarokh.
\newblock A survey of error-correcting codes for channels with symbol
  synchronization errors.
\newblock {\em IEEE Communications Surveys \& Tutorials}, 12(1), 2010.

\bibitem{moser2010constructive}
Robin~A Moser and G{\'a}bor Tardos.
\newblock A constructive proof of the general lov{\'a}sz local lemma.
\newblock {\em Journal of the ACM (JACM)}, 57(2):11, 2010.

\bibitem{shelton1981aperiodic}
Robert Shelton.
\newblock Aperiodic words on three symbols.
\newblock {\em Journal f{\"u}r die Reine und Angewandte Mathematik},
  321:195--209, 1981.

\bibitem{shelton1982aperiodic}
Robert~O Shelton and Raj~P Soni.
\newblock Aperiodic words on three symbols. iii.
\newblock {\em Journal f{\"u}r die Reine und Angewandte Mathematik},
  330:44--52, 1982.

\bibitem{thue1977unendliche}
A~Thue.
\newblock {\"U}ber unendliche zeichenreihen (1906).
\newblock {\em Selected Mathematical Papers of Axel Thue.
  Universitetsforlaget}, 1977.

\bibitem{thue1912gegenseitige}
Axel Thue.
\newblock {\em {\"U}ber die gegenseitige Lage gleicher Teile gewisser
  Zeichenreihen, von Axel Thue...}
\newblock J. Dybwad, 1912.

\bibitem{zolotov2015another}
Boris Zolotov.
\newblock Another solution to the thue problem of non-repeating words.
\newblock {\em arXiv preprint arXiv:1505.00019}, 2015.

\end{thebibliography}
